\documentclass[journal, headsepline]{IEEEtran}

\usepackage{amsmath,amsfonts,amssymb,amsthm}
\usepackage{mathrsfs}
\usepackage{comment,blkarray}
\usepackage{multirow,bigdelim}
\usepackage{cite}
\usepackage[ruled,commentsnumbered, vlined]{algorithm2e}
\usepackage{tikz}
\usetikzlibrary{arrows,automata}
\usepackage[latin1]{inputenc}
\usepackage{verbatim}
\usepackage{graphicx}
\usepackage{cases}
\usepackage{booktabs}
\usepackage{caption}												
\usepackage{subcaption}	
\usepackage{etoolbox}

\usepackage{enumitem}
\usepackage{mathtools}

\makeatletter

\newtheorem{thm}{Theorem}
\newtheorem{prop}{Proposition}
\theoremstyle{definition}
\newtheorem{example}{Example}
\newtheorem{remark}{Remark}
\allowdisplaybreaks

\begin{document}

\title{Moving Target SAR Imaging Using Planar Arrays And Multidimensional Chinese Remainder Theorem (MD-CRT)--Part I: A General Framework}

\author{Guangpu Guo, \IEEEmembership{Graduate Student Member}, \IEEEmembership{IEEE}, and Xiang-Gen Xia, \IEEEmembership{Fellow}, \IEEEmembership{IEEE} 
       
\thanks{G. Guo and X.-G. Xia are with the Department of Electrical and Computer Engineering,
  University of Delaware, Newark, DE 19716, USA
  (e-mails: guangpu@udel.edu and xxia@ee.udel.edu).
  This work was supported in part by the National Science Foundation (NSF) under Grant CCF-2246917.
}
}

\maketitle

\begin{abstract}
In this two-part paper, we investigate synthetic aperture radar (SAR) moving target imaging using planar antenna arrays. For a target moving over a three-dimensional terrain, its accurate localization requires the joint estimation of the motion-induced cross-range shift and the target height. In Part I of this two-part paper, starting from the planar array imaging geometry and the corresponding signal model, we show that these two quantities can be unified into a two-dimensional parameter vector and represented, after two-dimensional discrete Fourier transform (2D-DFT) processing across the planar array, through a natural vector remainder formulation.
We first develop a general 2D-DFT matrix modulus framework and show that, in the two-dimensional setting, the associated 2D-DFT matrix modulus affects the propagation of vector remainder errors. Under a fixed array geometry and antenna number constraint, we derive an optimal construction of this matrix modulus and adopt it in the subsequent analysis. Under this construction, a single planar array provides only a folded estimate when the true parameter vector lies outside its unambiguous range. To resolve this ambiguity, we develop a multi-subarray framework in which multiple planar subarrays generate multiple vector remainders with different matrix moduli, and the desired parameter vector is recovered through the multidimensional Chinese remainder theorem (MD-CRT). To account for practical errors introduced by 2D-DFT quantization and additive noise, we further introduce an approximate 2D-DFT peak model for non-integer frequency vectors, incorporate robust MD-CRT, and establish sufficient conditions together with explicit reconstruction error bounds for both noiseless and noisy settings. Numerical results verify that the proposed multi-subarray framework enlarges the unambiguous range compared with a single planar array. The formulation developed in this paper also provides the basis for planar array design and performance analysis, which are further investigated in Part II of this two-part paper.
\end{abstract}



\section{Introduction}\label{s1}

Synthetic aperture radar (SAR) is widely used for terrain and moving target imaging. Compared with stationary targets, moving targets introduce additional phase terms into the received signals, which usually cause a displacement of a target in the azimuth (cross-range) direction of the SAR image. Accurate estimation of this motion-induced shift is therefore essential for moving target localization.

To address this problem, many moving target localization methods have been developed in the past. The early studies mainly considered single-channel SAR systems, in which moving target parameters are estimated from azimuth phase histories \cite{Werness1990,Barbarossa1992,Perry1999}. However, when the Doppler centroid exceeds the pulse repetition frequency (PRF), velocity ambiguity arises and the target position can no longer be uniquely determined \cite{Kirsch2003,Marques2005,Wang2006DualSpeed}. To improve localization performance, SAR systems with multiple receivers or antenna arrays have been studied, including space-time-frequency processing \cite{Barbarossa1994}, multi-channel SAR \cite{Ender1996}, and velocity SAR (VSAR) based on linear arrays \cite{Friedlander1997}.
In the VSAR framework, multiple complex SAR images are formed from different antenna elements, and the cross-range shift is estimated through a discrete Fourier transform (DFT) across the array. Because of the inherent $2\pi$ phase folding in the DFT, a large cross-range shift caused by a fast moving target may exceed the unambiguous range and lead to ambiguity. To deal with this problem, multi-frequency VSAR, dual-speed VSAR, and VSAR with nonuniform linear arrays have been proposed \cite{Wang2004MultiFreq,Li2007DualSpeed,gangli}. These methods exploit multiple folded observations with different moduli and apply the robust Chinese remainder theorem (CRT) \cite{robustCRT,radeee} to recover the shift over an enlarged unambiguous range.

The above methods are mainly designed for targets on the ground surface, where the main concern is the motion-induced cross-range shift. For targets moving over a three-dimensional terrain (or in air), however, the imaging geometry becomes more complex. In this case, accurate localization of a target requires not only the estimation of the cross-range shift, but also the estimation of the target height.
Conventional VSAR systems employ one-dimensional antenna arrays and therefore provide spatial sampling only along the azimuth direction. As a result, they cannot resolve the target height. One possible extension is to introduce an additional vertical antenna and estimate the height from interferometric phase \cite{xiaoweili}. In this way, the cross-range shift and the target height are obtained from observations associated with different array directions separately. 

In this paper, we consider a more general target motion model. Instead of assuming that targets move on the ground or at a fixed height, we study targets moving over an arbitrary terrain (or in air), such as mountainous environments, with three-dimensional velocities along arbitrary directions, as considered in \cite{WasSAR_3D_Trajectory}. Under this setting, accurate localization of a target requires the joint estimation of the motion-induced cross-range shift and the target height. To address this problem, we adopt a planar antenna array, which has been explored in recent studies, for example, \cite{Incremental3DSAR,FDA_MIMO_Planar}, and formulate the two quantities in a unified two-dimensional vector form. The key difference from linear antenna arrays is that a planar array provides spatial sampling in two dimensions and therefore matches the two-dimensional parameter structure of the problem well. This makes it possible to process the data jointly through a two-dimensional discrete Fourier transform (2D-DFT) \cite{ar}.

To establish this formulation, we start from the three-dimensional imaging geometry and the corresponding signal model, which follows the standard VSAR framework but is adapted to the planar array and three-dimensional-motion setting. Based on this model, we show that, after 2D-DFT processing across the planar array, the motion-induced cross-range shift and the target height admit a two-dimensional vector remainder formulation under a general 2D-DFT matrix modulus. This is different from the one-dimensional case. For a one-dimensional linear array, the associated DFT modulus reduces to a scalar. Thus, under a given antenna number constraint there is no additional structural freedom to choose, i.e., there is only one possibility for an absolute value of a scalar. In contrast, in the two-dimensional setting, the modulus becomes a matrix, and there are many choices of different matrices for a given absolute determinant value. In the meanwhile, the matrix structure directly affects how vector remainder errors propagate to the desired parameter vector. As a result, even for planar arrays with the same number of antenna elements, i.e., the 2D-DFT matrix modulus with the same absolute determinant value, different matrix structures may exhibit different robustness behavior. We therefore further study the choice of the associated 2D-DFT matrix modulus and derive an optimal construction, for a fixed array geometry and under a fixed antenna number constraint, that minimizes the error amplification factor. In the subsequent development, we adopt this optimal construction throughout.

As we shall see later, a single planar array may suffer from folding. When the true parameter vector lies outside a range, called \textit{unambiguous range}, the 2D-DFT yields only a folded estimate, namely, a vector remainder of the true parameter vector modulo a certain matrix. This leads to the ambiguities for both parameters in the parameter vector. To enlarge the unambiguous range, we employ multiple planar subarrays that generate multiple folded observations with different matrix moduli, and then combine them through the multidimensional Chinese remainder theorem (MD-CRT) \cite{MD1,gMDCRT}. This can be viewed as a two-dimensional extension of the idea in \cite{gangli}, in which multiple linear subarrays generate multiple scalar remainders that are then combined to resolve ambiguity by using the conventional 1D CRT. Comparing to the 1D linear arrays with scalars in~\cite{gangli}, 2D planar arrays with matrices studied in this two-part paper have much more degrees of freedom in the design and are much more challenging.

In practice, the 2D-DFT is evaluated on a discrete frequency grid, so the detected vector remainders are affected by quantization. In addition, additive noise may further perturb the detected vector remainders. Unlike the one-dimensional case, the quantization effect in two dimensions is not completely characterized by simple nearest grid rounding on the original coordinates, because for non-integer frequency vectors the exact 2D-DFT peak selection regions are generally skewed by the 2D-DFT matrix modulus. We therefore introduce an approximate 2D-DFT peak model for non-integer frequency vectors, which replaces the true frequency vector by a nearby integer vector and represents the detected peak through the corresponding standard vector remainder. This provides a tractable bridge from practical 2D-DFT peak detection to the subsequent MD-CRT reconstruction analysis. Under these errors, the classical MD-CRT is no longer sufficient, because even small vector remainder errors may cause a large reconstruction error. This makes robust MD-CRT necessary for robust recovery \cite{MD2,mstage-mdcrt}. Based on the proposed multi-subarray framework, we then analyze robust recovery in both noiseless (quantization errors) and noisy (additive noise in a signal) settings, and establish sufficient conditions together with explicit reconstruction error bounds for robust recovery. Particular planar subarray designs and simulation results will be presented in Part II~\cite{part2} of this two-part paper.

The remainder of this paper is organized as follows. Section~\ref{s2} briefly reviews the necessary notations and concepts on integer matrices and the MD-CRT used in this paper. Section~\ref{s3} introduces the imaging geometry and signal model, establishes the 2D-DFT-based vector remainder formulation, and studies the associated 2D-DFT matrix modulus design for a given array geometry. Section~\ref{s4} analyzes the ambiguity problem, presents an approximate 2D-DFT peak model for non-integer frequency vectors, and develops the robust recovery framework based on multiple planar subarrays. Section~\ref{s5} presents some numerical results. Finally, Section~\ref{s6} concludes the paper.

\section{Preliminaries}\label{s2}

In this section, we briefly review several basic concepts on two-dimensional integer matrices that are needed for the MD-CRT framework used later. For further background on integer matrices and related results of the MD-CRT, we refer to~\cite{MD1,gMDCRT,MD2,mstage-mdcrt,guo25,primematrix,matrix}.

Throughout this section, $\mathbb{Z}$ and $\mathbb{R}$ denote the sets of integers and real numbers, respectively. Bold lowercase letters such as $\mathbf n$, $\mathbf f$, and $\mathbf r$ denote two-dimensional vectors, and bold uppercase letters such as $\mathbf M$, $\mathbf N$, and $\mathbf P$ denote $2\times2$ matrices, unless stated otherwise. The identity matrix is denoted by $\mathbf I$, the determinant of a matrix $\mathbf M$ is denoted by $\det(\mathbf M)$, the adjugate matrix is denoted by $\operatorname{adj}(\mathbf M)$, and $^{\top}$ denotes the transpose of a vector or matrix. For a vector $\mathbf a=[a_1,a_2]^\top\in\mathbb R^2$, $\|\mathbf a\|_\infty \triangleq \max\{|a_1|,|a_2|\}$, while, for a matrix $\mathbf A=[A_{ij}]\in\mathbb R^{2\times2}$, $\|\mathbf A\|_\infty \triangleq \max_{1\le i\le 2}\sum_{j=1}^2 |A_{ij}|$.

\begin{enumerate}[leftmargin=1.3em]

\item \textbf{Lattice:}
Given a nonsingular matrix $\mathbf M\in\mathbb R^{2\times2}$, the lattice generated by $\mathbf M$ is defined as
\[
\mathcal L(\mathbf M)=\{\mathbf M\mathbf n\mid \mathbf n\in\mathbb Z^2\}.
\]

\item \textbf{Fundamental parallelepiped (FPD):}
For any nonsingular matrix $\mathbf M\in\mathbb R^{2\times2}$, define
\[
\mathcal P(\mathbf M)
=
\left\{
\mathbf r\in\mathbb R^2
\mid
\mathbf r=\mathbf M\mathbf x,\ \mathbf x\in[0,1)^2
\right\}.
\]
If $\mathbf M\in\mathbb Z^{2\times2}$ is a nonsingular integer matrix, we further define the fundamental parallelepiped (FPD) of $\mathbf M$ as
\[
\mathcal N(\mathbf M)
\triangleq
\mathcal P(\mathbf M)\cap\mathbb Z^2.
\]
That is, $\mathcal N(\mathbf M)$ is the set of integer vectors contained in $\mathcal P(\mathbf M)$. It contains exactly $|\det(\mathbf M)|$ elements.

\item \textbf{Vector remainder representation:}
Let $\mathbf M\in\mathbb R^{2\times2}$ be a nonsingular matrix. Then any vector $\mathbf f\in\mathbb R^2$ can be uniquely written as
\[
\mathbf f=\mathbf M\mathbf n+\mathbf r,
\qquad
\mathbf n\in\mathbb Z^2,\ \mathbf r\in\mathcal P(\mathbf M).
\]
We write $\mathbf f\equiv\mathbf r\ \mathrm{mod}\ \mathbf M$,
and call $\mathbf r$ the \emph{vector remainder} of $\mathbf f$ modulo $\mathbf M$.

When $\mathbf M\in\mathbb Z^{2\times2}$ and $\mathbf f\in\mathbb Z^2$, the vector remainder $\mathbf r$ is also integer-valued, and hence $\mathbf r\in\mathcal N(\mathbf M)$.
Therefore, any integer vector $\mathbf f\in\mathbb Z^2$ admits the unique decomposition for $\mathbf M\in \mathbb{Z}^{2\times2}$ as
\[
\mathbf f=\mathbf M\mathbf n+\mathbf r,
\qquad
\mathbf n\in\mathbb Z^2,\ \mathbf r\in\mathcal N(\mathbf M).
\]
This integer vector remainder representation is the one used in the MD-CRT framework. The real-valued version above will be useful later when we analyze non-integer frequency vectors and quantization effects.

\item \textbf{Unimodular matrix:}
An integer matrix $\mathbf U\in\mathbb Z^{2\times2}$ is said to be unimodular if $|\det(\mathbf U)|=1$.

\item \textbf{Greatest common left divisor (gcld):}
Given two integer matrices $\mathbf M$ and $\mathbf N$, a nonsingular integer matrix $\mathbf G$ is called a common left divisor (cld) of $\mathbf M$ and $\mathbf N$ if both $\mathbf G^{-1}\mathbf M$ and $\mathbf G^{-1}\mathbf N$ are integer matrices. If every cld of $\mathbf M$ and $\mathbf N$ is also a left divisor of $\mathbf G$, then $\mathbf G$ is called the greatest common left divisor (gcld) of $\mathbf M$ and $\mathbf N$.

\item \textbf{Least common right multiple (lcrm):}
An integer matrix $\mathbf R$ is called a common right multiple (crm) of integer matrices $\mathbf M$ and $\mathbf N$ if there exist integer matrices $\mathbf P$ and $\mathbf Q$ such that
$\mathbf R=\mathbf M\mathbf P=\mathbf N\mathbf Q$.
If every crm of $\mathbf M$ and $\mathbf N$ is also a right multiple of $\mathbf R$, then $\mathbf R$ is called a least common right multiple (lcrm) of $\mathbf M$ and $\mathbf N$. Although the lcrm is not unique, the absolute determinant value $|\det(\mathbf R)|$ is uniquely determined.

\end{enumerate}

\section{Problem Setup}\label{s3}

In this section, we introduce the imaging geometry and the corresponding signal model for a moving target observed by a planar array in a SAR system. Based on this setup, we derive the antenna-dependent image expression and show that, across the planar array, the images admit a two-dimensional complex exponential structure. We then reformulate this structure in a general lattice-based 2D-DFT framework, establish the corresponding two-dimensional vector remainder formulation for the motion-induced cross-range shift and the target height, and analyze how the choice of the associated 2D-DFT matrix modulus affects the propagation of the vector remainder error. Under a fixed array geometry and antenna number constraint, we further derive an optimal construction of this matrix modulus. Finally, we illustrate the resulting planar array structure through an example.

\subsection{Geometry Model}\label{subsec:model}

\begin{figure}[htbp]
  \centering
  \includegraphics[width=0.7\linewidth]{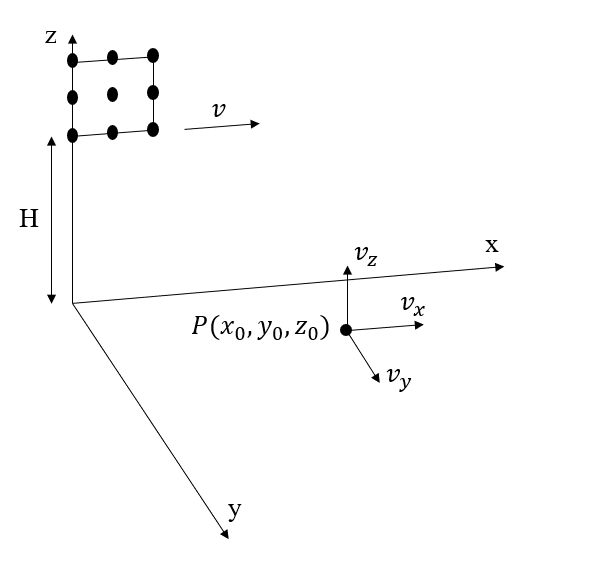}
  \caption{Geometry of the considered SAR imaging model.}
  \label{fig:model}
\end{figure}

The geometry is shown in Fig.~\ref{fig:model}.
The $x$-axis denotes the cross-range direction, the $y$-axis denotes the range direction,
and the $z$-axis denotes the vertical direction.

Consider a SAR platform moving with a constant velocity $v$
along the positive $x$-axis at a fixed altitude $H$.
A planar antenna array is placed on the platform.
All antenna elements lie in the $x$--$z$ plane with fixed positions and move together with the platform.

The antenna locations are modeled as a two-dimensional lattice generated by a matrix $\mathbf{M}\in \mathbb{R}^{2\times2}$,
where the entries of $\mathbf{M}$ are not necessarily integers, allowing for more general planar array designs.
This matrix $\mathbf{M}$ is called an array generating matrix or array geometry.
Each antenna element is indexed by an integer vector $\mathbf {u}\in\mathbb {Z}^2$,
and its planar coordinate is given by
\begin{equation}\label{array_geometry}
\mathbf M\mathbf u =
\begin{bmatrix}
x_{\mathbf u}\\
z_{\mathbf u}
\end{bmatrix}.
\end{equation}
At slow time $t=0$, the position of the $\mathbf u$-th antenna element is
\begin{equation}\notag
\mathbf r_{\mathbf u}(0)
=
[x_{\mathbf u},\,0,\,z_{\mathbf u}+H]^{\top}.
\end{equation}
Due to the platform motion, the antenna position at slow time $t$ becomes
\begin{equation}\notag
\mathbf r_{\mathbf u}(t)
=
[x_{\mathbf u}+vt,\,0,\,z_{\mathbf u}+H]^{\top}.
\end{equation}

We consider a single moving point target, denoted by $P$.
The initial position of the target is denoted by
\begin{equation}\notag
\mathbf r_P(0)=[x_0,\,y_0,\,z_0]^{\top},
\end{equation}
and the target moves with a constant velocity
\begin{equation}\notag
\mathbf v_P=[v_x,\,v_y,\,v_z]^{\top},
\end{equation}
as illustrated in Fig.~\ref{fig:model}.
Accordingly, the target position at slow time $t$ is given by
\begin{equation}\notag
\mathbf r_P(t)
=
[x_0+v_x t,\,y_0+v_y t,\,z_0+v_z t]^{\top}.
\end{equation}

\subsection{Signal Model and Basic Processing}

The $\mathbf{0}$-th antenna element serves as both transmitter and receiver, while all the other antenna elements act as receivers. Suppose that the $\mathbf{0}$-th antenna element transmits a linear frequency modulated (LFM) signal with carrier frequency $f_c$.

We next derive the signal model and the corresponding processed images. The processing follows the standard VSAR framework \cite{Friedlander1997,Wang2004MultiFreq,Li2007DualSpeed,gangli}, but the geometry is different: the target is allowed to move in three dimensions, and the antenna elements have nonzero height coordinates. As a result, additional phase terms appear.
For clarity, we briefly outline the key steps below and highlight the key differences.

Let $R_{\mathbf u}(t)$ denote the distance (range) between the $\mathbf u$-th antenna element and the moving target at slow time $t$. 
Under the far-field and small-aperture assumptions, $R_{\mathbf u}(t)$ can be approximated by a second-order Taylor expansion around the reference range $R=R_{\mathbf 0}(0)$, i.e., the distance between the target and the transmitter at time $t=0$.
Following the standard VSAR processing model \cite{Friedlander1997} after quadratic phase correction and subsequent range compression, the received signal phase at the $\mathbf u$-th antenna can be expressed as
\begin{equation}\label{eq:phase}
\phi_{\mathbf u}(t)
=
-\frac{2\pi}{\lambda}\big(R_{\mathbf 0}(t)+R_{\mathbf u}(t)\big),
\end{equation}
where $\lambda$ denotes the carrier wavelength.
Although the resulting phase expression has the same form as that in conventional VSAR models, the reference range $R$ and the distance $R_{\mathbf u}(t)$ here are defined based on the three-dimensional target motion and the planar array geometry.

Starting from~\eqref{eq:phase}, the second-order expansions of $R_{\mathbf 0}(t)$ and $R_{\mathbf u}(t)$ give
\begin{equation}\label{eq:phase_expand}
\begin{aligned}
\phi_{\mathbf u}(t)
\approx\;
-&\frac{4\pi R}{\lambda} -\frac{2\pi}{R\lambda}
\Big(
(v_x^2+v_y^2+v_z^2-2vv_x)t^2
\Big) \\
&-\frac{2\pi}{R\lambda}
\Big(
2(x_0v_x-x_0v+v_yy_0-Hv_z+z_0v_z)t
\Big) \\
&+\frac{2\pi}{R\lambda}
(x_0x_{\mathbf u}+z_0z_{\mathbf u})
+\zeta_{1,\mathbf u}(t)+\zeta_{2,\mathbf u}(t),
\end{aligned}
\end{equation}
where
\begin{equation}\label{eq:zeta1}
\zeta_{1,\mathbf u}(t)
=
-\frac{\pi}{R\lambda}
\big(
2v^2t^2+2x_{\mathbf u}vt+x_{\mathbf u}^2+z_{\mathbf u}^2+2Hz_{\mathbf u}
\big)
\end{equation}
contains only known platform parameters and antenna coordinates, and hence can be compensated in advance, and
\begin{equation}\label{eq:zeta2}
\zeta_{2,\mathbf u}(t)
=
\frac{2\pi}{R\lambda}
(v_xx_{\mathbf u}+v_zz_{\mathbf u})t
\end{equation}
is a coupling term between the target velocity and the array aperture.
Under the far-field and small-aperture assumptions, $|x_{\mathbf u}|,|z_{\mathbf u}|\ll R$, so $\zeta_{2,\mathbf u}(t)$ is small and can be neglected.
After these simplifications, the phase at the $\mathbf u$-th antenna can be approximated as
\begin{equation}\label{eq:phase_simplified}
\begin{aligned}
\phi_{\mathbf u}(t)
\approx
-&\frac{4\pi R}{\lambda}
-\frac{2\pi}{R\lambda}
\Big(
(v_x^2+v_y^2+v_z^2-2vv_x)t^2
+2(x_0v_x\\&-x_0v
+v_yy_0-Hv_z+z_0v_z)t
-x_0x_{\mathbf u}-z_0z_{\mathbf u}
\Big).
\end{aligned}
\end{equation}

As in the conventional VSAR processing \cite{Friedlander1997,Wang2004MultiFreq,Li2007DualSpeed,gangli}, the quadratic term in \eqref{eq:phase_simplified} is first estimated from the slow-time data using joint time-frequency analysis techniques and then compensated. 
We then take the Fourier transform, with respect to the slow time $t$, of the compensated version of $\exp(j\phi_{\mathbf u}(t))$ and obtain the complex image at the $\mathbf u$-th antenna,
\begin{equation}\label{eq:Su_initial}
\begin{aligned}
S_{\mathbf u}(n)
&=
\exp\!\left(
-j\frac{4\pi R}{\lambda}
+j\frac{2\pi}{R\lambda}(x_0x_{\mathbf u}+z_0z_{\mathbf u})
\right)
\\
&\quad\times
\delta\!\big(n-n_{x_0}-\Delta_{\rm shift}\big),
\end{aligned}
\end{equation}
where $n$ is the discrete cross-range index in the image and $\delta(\cdot)$ denotes the discrete delta function.

Here, $n_{x_0}$ is the discrete index corresponding to the true cross-range location $x_0$ of the moving target, namely,
\begin{equation}
x_0 = \Delta_x\, n_{x_0},
\end{equation}
where $\Delta_x$ denotes the cross-range resolution.
The quantity $\Delta_{\rm shift}$ denotes the motion-induced shift of the target from the true position $n_{x_0}$ in the SAR image, and is given by
\begin{equation}\label{eq:Deltashift}
\Delta_{\rm shift}
=
-\frac{x_0 v_x + y_0 v_y + z_0 v_z - H v_z}{v\,\Delta_x}.
\end{equation}
Once $\Delta_{\rm shift}$ is estimated, the true cross-range position $x_0$ can be readily recovered.

When $v_x=v_y=v_z=0$, the shift term $\Delta_{\rm shift}$ becomes zero, and~\eqref{eq:Su_initial} reduces to the stationary target case.

\subsection{2D-DFT Formulation}\label{subsec:md_dft}

From~\eqref{eq:Su_initial}, the image peak appears at the index
$n_{x_0}+\Delta_{\rm shift}$, which can be directly detected from the SAR image.
We therefore compensate the corresponding cross-range phase term by multiplying
\eqref{eq:Su_initial} with
\begin{equation}\label{eq:phase_comp}
\exp\!\left(
-j\frac{2\pi}{R\lambda}
\big(n_{x_0}+\Delta_{\rm shift}\big)\Delta_x\,x_{\mathbf u}
\right).
\end{equation}
Using $x_0=\Delta_x n_{x_0}$, the compensated image becomes
\begin{equation}\label{eq:Su_comp}
\begin{aligned}
S_{\mathbf u}(n)
&=
\exp\!\Big(
-j\frac{4\pi R}{\lambda}
+j\frac{2\pi}{R\lambda}
\big(
-\Delta_{\rm shift}\Delta_x\,x_{\mathbf u}
+z_0 z_{\mathbf u}
\big)
\Big)
\\
&\quad\times
\delta\!\big(
n-n_{x_0}-\Delta_{\rm shift}
\big).
\end{aligned}
\end{equation}
For convenience, define the antenna-independent phase term
\begin{equation}
C_0
=
\exp\!\left(
-j\frac{4\pi R}{\lambda}
\right),
\end{equation}
which does not affect the subsequent spatial processing across the planar array.
Then~\eqref{eq:Su_comp} can be rewritten as
\begin{equation}\label{eq:Su_comp_clean}
\begin{aligned}
S_{\mathbf u}(n)
&=
C_0
\exp\!\left(
j\frac{2\pi}{R\lambda}
\big(
-\Delta_{\rm shift}\Delta_x\,x_{\mathbf u}
+z_0 z_{\mathbf u}
\big)
\right)
\\
&\quad\times
\delta\!\big(n-n_{x_0}-\Delta_{\rm shift}\big).
\end{aligned}
\end{equation}

We now introduce the parameter vector
\[
\mathbf g
\triangleq
\begin{bmatrix}
-\Delta_{\rm shift}\Delta_x\\
z_0
\end{bmatrix},
\]
which jointly characterizes the motion-induced cross-range shift and the target height.
Recovering $\mathbf g$ is therefore equivalent to jointly estimating $\Delta_{\rm shift}$ and $z_0$.
Using the lattice representation of the antenna positions,
\[
\begin{bmatrix}
x_{\mathbf u}\\
z_{\mathbf u}
\end{bmatrix}
=
\mathbf M\mathbf u,
\]
the antenna-dependent phase terms in~\eqref{eq:Su_comp_clean} can be written as
\begin{equation}
-\Delta_{\rm shift}\Delta_x\,x_{\mathbf u}+z_0 z_{\mathbf u}
=
\mathbf g^\top \mathbf M\mathbf u
=
(\mathbf M^\top\mathbf g)^\top\mathbf u.
\end{equation}
Therefore, across the planar array, the image follows a two-dimensional complex exponential structure with respect to the antenna index $\mathbf u$, and can be expressed as
\begin{equation}\label{eq:Su_exp_structure}
S_{\mathbf u}(n)
=
C_0\,
\exp\!\left(
j\frac{2\pi}{R\lambda}
(\mathbf M^\top\mathbf g)^\top\mathbf u
\right)
\delta\!\big(n-n_{x_0}-\Delta_{\rm shift}\big).
\end{equation}
This structure suggests that the parameter vector $\mathbf g$ can be efficiently extracted through a 2D-DFT across the antenna indices $\mathbf{u}$.

To proceed, we first develop a general lattice-based 2D-DFT formulation for an arbitrary integer matrix modulus~\cite{ar}.
Define
\begin{equation}\label{eq:p_from_g_generalL_rev}
\mathbf p
=
\frac{1}{R\lambda}\,\mathbf M^\top\mathbf g
\in\mathbb R^2.
\end{equation}
Then~\eqref{eq:Su_exp_structure} can be rewritten as
\begin{equation}\label{eq:single_exp_signal_generalL}
S_{\mathbf u}(n)
=
C_0\,
\exp\!\left(
j2\pi\,\mathbf p^\top\mathbf u
\right)
\delta\!\big(n-n_{x_0}-\Delta_{\rm shift}\big).
\end{equation}
Let $\mathbf N\in\mathbb Z^{2\times2}$ be a nonsingular integer matrix.
To exploit the orthogonality over a finite lattice, we sample the antenna index on one FPD of $\mathbf N^\top$, i.e., $\mathcal N(\mathbf N^\top)$, and define the 2D-DFT over $\mathcal N(\mathbf N^\top)$, for $\mathbf k\in\mathcal N(\mathbf N)$, by
\begin{equation}\label{eq:mddft_def_generalL_rev}
\tilde S(\mathbf k,n)
=
\sum_{\mathbf u\in\mathcal N(\mathbf N^\top)}
S_{\mathbf u}(n)\,
\exp\!\left(
-j2\pi\,\mathbf k^\top\mathbf N^{-\top}\mathbf u
\right).
\end{equation}
Substituting~\eqref{eq:single_exp_signal_generalL} into~\eqref{eq:mddft_def_generalL_rev} gives
\begin{equation}\label{eq:mddft_expand_generalL_1}
\begin{aligned}
\tilde S(\mathbf k,n)
&=
C_0\,
\delta\!\big(n-n_{x_0}-\Delta_{\rm shift}\big)
\\
&\quad\times
\sum_{\mathbf u\in\mathcal N(\mathbf N^\top)}
\exp\!\left(j2\pi\,\mathbf p^\top\mathbf u\right)
\exp\!\left(-j2\pi\,\mathbf k^\top\mathbf N^{-\top}\mathbf u\right).
\end{aligned}
\end{equation}

Define
\begin{equation}\label{eq:q_def_generalL}
\mathbf q
\triangleq
\mathbf N\mathbf p
=
\frac{1}{R\lambda}\,\mathbf N\mathbf M^\top\mathbf g.
\end{equation}
Since
\[
\mathbf p^\top\mathbf u
=
(\mathbf N^{-1}\mathbf q)^\top\mathbf u
=
\mathbf q^\top\mathbf N^{-\top}\mathbf u,
\]
the inner summation in the right hand side of~\eqref{eq:mddft_expand_generalL_1} can be rewritten as
\begin{equation}\label{eq:mddft_expand_generalL_2}
\begin{aligned}
&
\sum_{\mathbf u\in\mathcal N(\mathbf N^\top)}
\exp\!\left(j2\pi\,\mathbf q^\top\mathbf N^{-\top}\mathbf u\right)
\exp\!\left(-j2\pi\,\mathbf k^\top\mathbf N^{-\top}\mathbf u\right)
\\
=&
\sum_{\mathbf u\in\mathcal N(\mathbf N^\top)}
\exp\!\left(
j2\pi\,(\mathbf q-\mathbf k)^\top\mathbf N^{-\top}\mathbf u
\right).
\end{aligned}
\end{equation}

We first consider the ideal case where $\mathbf q$ in~\eqref{eq:q_def_generalL} is an integer vector.
Then the summation in~\eqref{eq:mddft_expand_generalL_2} is a finite sum of complex exponentials over $\mathcal N(\mathbf N^\top)$.
Let $\mathbf r\in\mathcal N(\mathbf N)$ be the vector remainder of $\mathbf q$ modulo $\mathbf N$.
When $\mathbf k=\mathbf r$, then $\mathbf q-\mathbf r=\mathbf N\mathbf z$ for some $\mathbf z\in\mathbb Z^2$, and hence
\[
(\mathbf q-\mathbf r)^\top\mathbf N^{-\top}\mathbf u
=
\mathbf z^\top\mathbf u
\in\mathbb Z
\]
for every $\mathbf u\in\mathbb Z^2$, so each complex exponential equals $1$.
When $\mathbf k\neq \mathbf r$, then the corresponding complex exponential is a nontrivial term on the finite lattice induced by $\mathbf N^\top$, and its sum over $\mathcal N(\mathbf N^\top)$ is zero.
Therefore,
\begin{equation}\label{eq:mddft_integer_case_generalL_rev}
\sum_{\mathbf u\in\mathcal N(\mathbf N^\top)}
\exp\!\left(
j2\pi\,(\mathbf q-\mathbf k)^\top\mathbf N^{-\top}\mathbf u
\right)
=
|\det(\mathbf N)|\,\delta(\mathbf k-\mathbf r),
\end{equation}
where $\mathbf r$ is the vector remainder of $\mathbf q$ modulo $\mathbf N$.
Substituting~\eqref{eq:mddft_integer_case_generalL_rev} into~\eqref{eq:mddft_expand_generalL_1}, we obtain
\begin{equation}\label{eq:2ddft_result_integer_case}
\tilde S(\mathbf k,n)
=
C_0\,|\det(\mathbf N)|
\,\delta(\mathbf k-\mathbf r)
\,\delta\!\big(n-n_{x_0}-\Delta_{\rm shift}\big),
\end{equation}
where $\mathbf r \equiv \mathbf q \ \mathrm{mod}\ \mathbf N$.

Therefore, after the 2D-DFT across the antenna index $\mathbf u$, the peak appears at $\mathbf k=\mathbf r$, from which the vector remainder $\mathbf r$ can be directly obtained.
Since $\mathbf q\equiv\mathbf r\ \mathrm{mod}\ \mathbf N$, it follows from~\eqref{eq:q_def_generalL} that
\begin{equation}
\frac{1}{R\lambda}\,\mathbf N\mathbf M^\top\mathbf g
\equiv
\mathbf r
\ \mathrm{mod}\ \mathbf N.
\end{equation}
Equivalently, there exists an integer vector $\mathbf n\in\mathbb Z^2$ such that
\begin{equation}
\frac{1}{R\lambda}\,\mathbf N\mathbf M^\top\mathbf g
=
\mathbf N\mathbf n+\mathbf r.
\end{equation}
Multiplying both sides by $R\lambda\,\mathbf M^{-\top}\mathbf N^{-1}$ yields
\begin{equation}
\mathbf g
=
R\lambda\,\mathbf M^{-\top}\mathbf n
+
R\lambda\,\mathbf M^{-\top}\mathbf N^{-1}\mathbf r.
\end{equation}
Hence,
\begin{equation}\label{eq:g_mod_relation}
\mathbf g
\equiv
R\lambda\,\mathbf M^{-\top}\mathbf N^{-1}\mathbf r
\ \mathrm{mod}\ R\lambda\,\mathbf M^{-\top}.
\end{equation}

We next consider the general case where $\mathbf q$ is not necessarily integer-valued.
Then the detected 2D-DFT peak may not be exactly the discrete vector remainder of an integer vector modulo $\mathbf N$.
This  may cause an error in the detected vector remainder $\mathbf r$ and thus an error in $\mathbf q$. The error in $\mathbf q$ will propagate to $\mathbf g$ through the matrix factor $\mathbf M^{-\top}\mathbf N^{-1}$ as shown in \eqref{eq:q_def_generalL}. This motivates our choice of $\mathbf N$ below.
To see this, let
\begin{equation}\label{eq:q_continuous_remainder}
\mathbf q
=
\mathbf N\mathbf n+\mathbf r_{\mathrm c},
\qquad
\mathbf r_{\mathrm c}\in\mathcal P(\mathbf N),\ \mathbf n\in\mathbb Z^2,
\end{equation}
where $\mathbf r_{\mathrm c}$ is the vector remainder of the non-integer vector $\mathbf q$ modulo $\mathbf N$.
From \eqref{eq:q_def_generalL}, the corresponding vector remainder of $\mathbf g$ is
\begin{equation}\label{eq:g_remainder_continuous}
\mathbf g_{\mathrm{rem}}
=
R\lambda\,\mathbf M^{-\top}\mathbf N^{-1}\mathbf r_{\mathrm c}.
\end{equation}
Suppose that the detected vector remainder from the 2D-DFT peak detection can be modeled as
\begin{equation}
\hat{\mathbf r}
=
\mathbf r_{\mathrm c}+\mathbf e_r,
\end{equation}
where $\mathbf e_r$ denotes the vector remainder estimation error.
Then the induced error in the vector remainder of $\mathbf g$ is
\begin{equation}\label{eq:g_error_exact_generalN}
\hat{\mathbf g}_{\mathrm{rem}}
-
\mathbf g_{\mathrm{rem}}
=
R\lambda\,\mathbf M^{-\top}\mathbf N^{-1}\mathbf e_r,
\end{equation}
and therefore
\begin{equation}\label{eq:g_error_bound_inf_generalN}
\left\|
\hat{\mathbf g}_{\mathrm{rem}}
-
\mathbf g_{\mathrm{rem}}
\right\|_\infty
\le
R\lambda\,
\left\|
\mathbf M^{-\top}\mathbf N^{-1}
\right\|_\infty
\,
\|\mathbf e_r\|_\infty.
\end{equation}
A more detailed treatment of the non-integer case and the associated quantization effect will be given in Section~\ref{s4.2}.

The bound in~\eqref{eq:g_error_bound_inf_generalN} shows that, for a given array geometry matrix $\mathbf M$, the choice of $\mathbf N$ directly affects the error amplification from the vector remainder error $\mathbf e_r$ to the parameter vector $\mathbf g$.
This motivates us to select $\mathbf N$ so that
\[
\left\|
\mathbf M^{-\top}\mathbf N^{-1}
\right\|_\infty
\]
is as small as possible.

In the meantime, the choice of $\mathbf N$ is also constrained by the array size.
This is because, from \eqref{eq:mddft_def_generalL_rev}, the antenna index $\mathbf u$ is sampled over the whole FPD set $\mathcal N(\mathbf N^\top)$, whose cardinality is
\[
|\mathcal N(\mathbf N^\top)|
=
|\det(\mathbf N)|.
\]
Therefore, $|\det(\mathbf N)|$ equals the number of antenna elements used in the corresponding planar array.
In practice, once the array geometry $\mathbf M$ is fixed, the number of available antenna elements is limited by the physical size of the platform.
It is thus natural to optimize $\mathbf N$ under a fixed absolute value of the determinant, when the array geometry $\mathbf M$ in \eqref{array_geometry} is fixed.

This design issue is specific to the two-dimensional setting considered here.
For a one-dimensional linear array, the corresponding $\mathbf{N}$ reduces to a scalar $N$.
In that case, $N$ represents both the number of points in the one-dimensional DFT and the number of antenna elements in the array.
Hence, under a given antenna number constraint, the choice of $N$ is unique and one simply uses the largest admissible value, and no additional structural design freedom exists.
In contrast, in the present two-dimensional case, $\mathbf N$ is a matrix, and when its absolute determinant value is fixed, it has many choices with different matrix structures that lead to different error amplification behavior.

This motivates the following optimization problem:
\begin{equation}\label{eq:opt_problem_fixed_det_2d}
\begin{aligned}
\min_{\mathbf N\in\mathbb Z^{2\times2}}
\quad &
\left\|
\mathbf M^{-\top}\mathbf N^{-1}
\right\|_\infty
\\
\text{subject to}\quad &
|\det(\mathbf N)|=l,
\end{aligned}
\end{equation}
where $l$ is a given positive integer representing the number of available antenna elements.
For convenience, define
\begin{equation}\label{eq:W_def_2d}
\mathbf W(\mathbf N)
\triangleq
\mathbf M^{-\top}\mathbf N^{-1}.
\end{equation}
The following theorem characterizes the optimal solution to \eqref{eq:opt_problem_fixed_det_2d}.

\begin{thm}\label{thm:optimal_N_2d}
Let $l$ be a fixed positive integer.
For the optimization problem in~\eqref{eq:opt_problem_fixed_det_2d}, every feasible integer matrix $\mathbf N$ satisfies
\begin{equation}\label{eq:det_lower_bound_inf_2d}
\|\mathbf W(\mathbf N)\|_\infty
\ge
\frac{1}{\sqrt{l\,|\det(\mathbf M)|}}.
\end{equation}
Moreover, if
\begin{equation}\label{eq:N_star_def_2d}
\mathbf N_\star
=
\alpha\,\operatorname{adj}(\mathbf M^\top)
\end{equation}
is integer-valued for some scalar $\alpha$ and satisfies
\begin{equation}\label{eq:l_alpha_relation}
l=\alpha^2|\det(\mathbf M)|,
\end{equation}
then $\mathbf N_\star$ attains the equality in~\eqref{eq:det_lower_bound_inf_2d} and is therefore a global minimizer of~\eqref{eq:opt_problem_fixed_det_2d}.
\end{thm}

\begin{proof}
Let $\mathbf N$ be any feasible integer matrix and write
\[
\mathbf W=\mathbf W(\mathbf N)=\mathbf M^{-\top}\mathbf N^{-1}.
\]
By Hadamard's determinant inequality~\cite{horn2012matrix} applied to the rows of $\mathbf W$, we have
\[
|\det(\mathbf W)|
\le
\prod_{i=1}^2 \|\operatorname{row}_i(\mathbf W)\|_2.
\]
Moreover, for each row vector,
\[
\|\operatorname{row}_i(\mathbf W)\|_2
\le
\|\operatorname{row}_i(\mathbf W)\|_1
\le
\|\mathbf W\|_\infty,
\]
where $\|\cdot\|_2$ and $\|\cdot\|_1$ denote the Euclidean norm and the $\ell_1$ norm of a vector, respectively. 
Therefore, we obtain
\begin{equation}\label{eq:W_det_bound}
\|\mathbf W\|_\infty
\ge
|\det(\mathbf W)|^{1/2}.
\end{equation}
Since
\[
\det(\mathbf W)
=
\det(\mathbf M^{-\top})\det(\mathbf N^{-1})
=
\frac{1}{\det(\mathbf M)\det(\mathbf N)},
\]
and every feasible $\mathbf N$ satisfies $|\det(\mathbf N)|=l$, we obtain
\begin{equation}\label{eq:W_det_exact}
|\det(\mathbf W)|
=
\frac{1}{l\,|\det(\mathbf M)|}.
\end{equation}
Combining \eqref{eq:W_det_bound} and \eqref{eq:W_det_exact}, we have
\[
\|\mathbf W(\mathbf N)\|_\infty
\ge
\frac{1}{\sqrt{l\,|\det(\mathbf M)|}},
\]
which proves~\eqref{eq:det_lower_bound_inf_2d}.

Now consider $\mathbf N_\star=\alpha\,\operatorname{adj}(\mathbf M^\top)$.
Using
\[
\operatorname{adj}(\mathbf M^\top)=\det(\mathbf M)\mathbf M^{-\top},
\]
we have
\[
\mathbf N_\star
=
\alpha\,\det(\mathbf M)\,\mathbf M^{-\top},
\]
and hence
\[
\mathbf W(\mathbf N_\star)
=
\mathbf M^{-\top}\mathbf N_\star^{-1}
=
\frac{1}{\alpha\,\det(\mathbf M)}\mathbf I.
\]
Therefore,
\begin{equation}\label{eq:W_star_norm}
\|\mathbf W(\mathbf N_\star)\|_\infty
=
\frac{1}{|\alpha|\,|\det(\mathbf M)|}.
\end{equation}
On the other hand,
\begin{equation}\label{eq:N_star_det}
|\det(\mathbf N_\star)|
=
\alpha^2\,|\det(\operatorname{adj}(\mathbf M^\top))|
=
\alpha^2\,|\det(\mathbf M)|.
\end{equation}
If~\eqref{eq:l_alpha_relation} holds, then \eqref{eq:N_star_det} gives
\[
|\det(\mathbf N_\star)|=l,
\]
so $\mathbf N_\star$ is feasible. Moreover, by~\eqref{eq:l_alpha_relation},
\[
\frac{1}{\sqrt{l\,|\det(\mathbf M)|}}
=
\frac{1}{\sqrt{\alpha^2|\det(\mathbf M)|^2}}
=
\frac{1}{|\alpha|\,|\det(\mathbf M)|}.
\]
Combining this with~\eqref{eq:W_star_norm}, we see that $\mathbf N_\star$ attains the lower bound in~\eqref{eq:det_lower_bound_inf_2d}. Therefore, $\mathbf N_\star$ is a global minimizer of~\eqref{eq:opt_problem_fixed_det_2d}.
\end{proof}

\begin{remark}
Although Theorem~\ref{thm:optimal_N_2d} is stated for the two-dimensional case considered in this paper, the same argument extends directly to an arbitrary dimension $D$.
\end{remark}

Theorem~\ref{thm:optimal_N_2d} shows that, for a given $\mathbf M$, the choice
\begin{equation}\label{eq:N_choice_final_2d}
\mathbf N
=
\alpha\,\operatorname{adj}(\mathbf M^\top)
\end{equation}
minimizes the error amplification factor
\[
\left\|
\mathbf M^{-\top}\mathbf N^{-1}
\right\|_\infty
\]
among all integer matrices with the same absolute determinant value, whenever $\mathbf N$ is integer-valued.

This result also has a simple intuitive explanation.
Since
\[
\operatorname{adj}(\mathbf M^\top)=\det(\mathbf M)\mathbf M^{-\top},
\]
choosing $\mathbf N$ proportional to $\operatorname{adj}(\mathbf M^\top)$ makes $\mathbf M^{-\top}\mathbf N^{-1}$
a scalar multiple of the identity matrix.
So, after this choice, the effect of the array geometry is the same in both coordinates.
In other words, the vector remainder error is balanced in the two directions, instead of being enlarged more in one direction than in the other.

The scalar $\alpha$ sets the scale of $\mathbf N$ in \eqref{eq:N_choice_final_2d}.
Since from \eqref{eq:mddft_def_generalL_rev}, one can see that the antenna indices need to be chosen from the whole set $\mathcal N(\mathbf N^\top)$, the number of antenna elements is $|\det(\mathbf N)|$.
Therefore, when the array geometry $\mathbf M$ is fixed, a larger feasible $|\alpha|$ means that more antenna elements need to be used.
Meanwhile, it also reduces the error amplification factor $\left\|
\mathbf M^{-\top}\mathbf N^{-1}
\right\|_\infty$.
So $\alpha$ can be viewed as a scaling factor that makes full use of the available antennas and also improves robustness.
On the other hand, all the antennas with indices in $\mathcal N(\mathbf N^{\top})$ need to be within a given platform of limited size. This means that when the array geometry $\mathbf M$ is fixed, we may use the scaling factor $\alpha$ to control the array occupied platform size.

Therefore, in the following, we adopt the choice in~\eqref{eq:N_choice_final_2d} and conduct all subsequent analysis under this setting.

Define
\begin{equation}\label{eq:Nandf}
\mathbf N
\triangleq
\alpha\,\operatorname{adj}(\mathbf M^{\top}),
\qquad
\mathbf f
\triangleq
\frac{\alpha\det(\mathbf M)}{R\lambda}\,\mathbf g,
\end{equation}
where $\alpha \in \mathbb{R}$ is chosen such that $\mathbf N$ is an integer matrix.
In general, $\mathbf f$ is not necessarily an integer vector.
In this subsection, we first consider the ideal case where $\mathbf f$ is also an integer vector.
The general non-integer case of $\mathbf f$, which introduces quantization errors, will be addressed in Section~\ref{s4.2}.

Since
\[
\operatorname{adj}(\mathbf M)\mathbf M
=
\det(\mathbf M)\mathbf I,
\]
we have
\begin{equation}\label{eq:mgandfn}
\frac{1}{R\lambda}\mathbf M^\top\mathbf g
=
\mathbf N^{-1}\mathbf f.
\end{equation}
Substituting \eqref{eq:mgandfn} into~\eqref{eq:Su_exp_structure}, we obtain
\begin{equation}\label{eq:Su_final}
S_{\mathbf u}(n)
=
C_0\,
\exp\!\left(
j2\pi\,\mathbf f^{\top}\mathbf N^{-\top}\mathbf u
\right)
\delta\!\big(n-n_{x_0}-\Delta_{\rm shift}\big).
\end{equation}

Applying the 2D-DFT with respect to $\mathbf u\in\mathcal N(\mathbf N^\top)$, and using the same derivation as in~\eqref{eq:mddft_def_generalL_rev}--\eqref{eq:2ddft_result_integer_case}, we obtain
\begin{equation}\label{eq:MDDFT_main}
\tilde S(\mathbf k,n)
=
C_0\,|\det(\mathbf N)|
\,\delta\!\big(\mathbf k-\mathbf r\big)
\,\delta\!\big(n-n_{x_0}-\Delta_{\rm shift}\big),
\end{equation}
where $\mathbf r\in\mathcal N(\mathbf N)$ is the vector remainder of $\mathbf f$ modulo $\mathbf N$, namely,
\begin{equation}\label{eq:f_congruence_main}
\mathbf f
\equiv
\mathbf r
\ \mathrm{mod}\ \mathbf N.
\end{equation}

Therefore, the 2D-DFT directly provides the vector remainder $\mathbf r$ of $\mathbf f$ modulo $\mathbf N$.
If $\mathbf f\in\mathcal N(\mathbf N)$, then $\mathbf r=\mathbf f$, and the parameter vector $\mathbf g$ can be uniquely recovered from~\eqref{eq:Nandf}.
Otherwise, only a folded version of $\mathbf f$, namely its vector remainder modulo $\mathbf N$, is obtained.
This ambiguity motivates the multi-subarray framework developed in the next section.

\subsection{Planar Array Structure}\label{subsec:array}

To apply the above 2D-DFT formulation in practice, two conditions are needed.
First, the matrix $\mathbf N$ must be an integer matrix.
This ensures that the complex exponential terms
$\exp\!\left(
j2\pi\,\mathbf f^{\top}\mathbf N^{-\top}\mathbf u
\right)$
are periodic with respect to the lattice generated by $\mathbf N^{\top}$.
Second, the antenna indices must fully cover the FPD of $\mathbf N^{\top}$.
That is, for each $\mathbf u\in\mathcal N(\mathbf N^\top)$, there must be an antenna element at the location $\mathbf M\mathbf u$.

We next give an example to illustrate how the antenna elements should be arranged so that these conditions are satisfied.

\begin{example}\label{ex:array}
Let
\begin{equation}\notag
\mathbf M
=
\begin{pmatrix}
1.5 & 0\\
-0.5 & 1
\end{pmatrix},
\end{equation}
and all antenna elements are placed on the lattice
generated by $\mathbf M$, i.e.,
each antenna location is of the form
$\mathbf M\mathbf u$ with $\mathbf u\in\mathbb Z^2$.

\begin{figure}[htbp]
  \centering
  \includegraphics[width=0.8\linewidth]{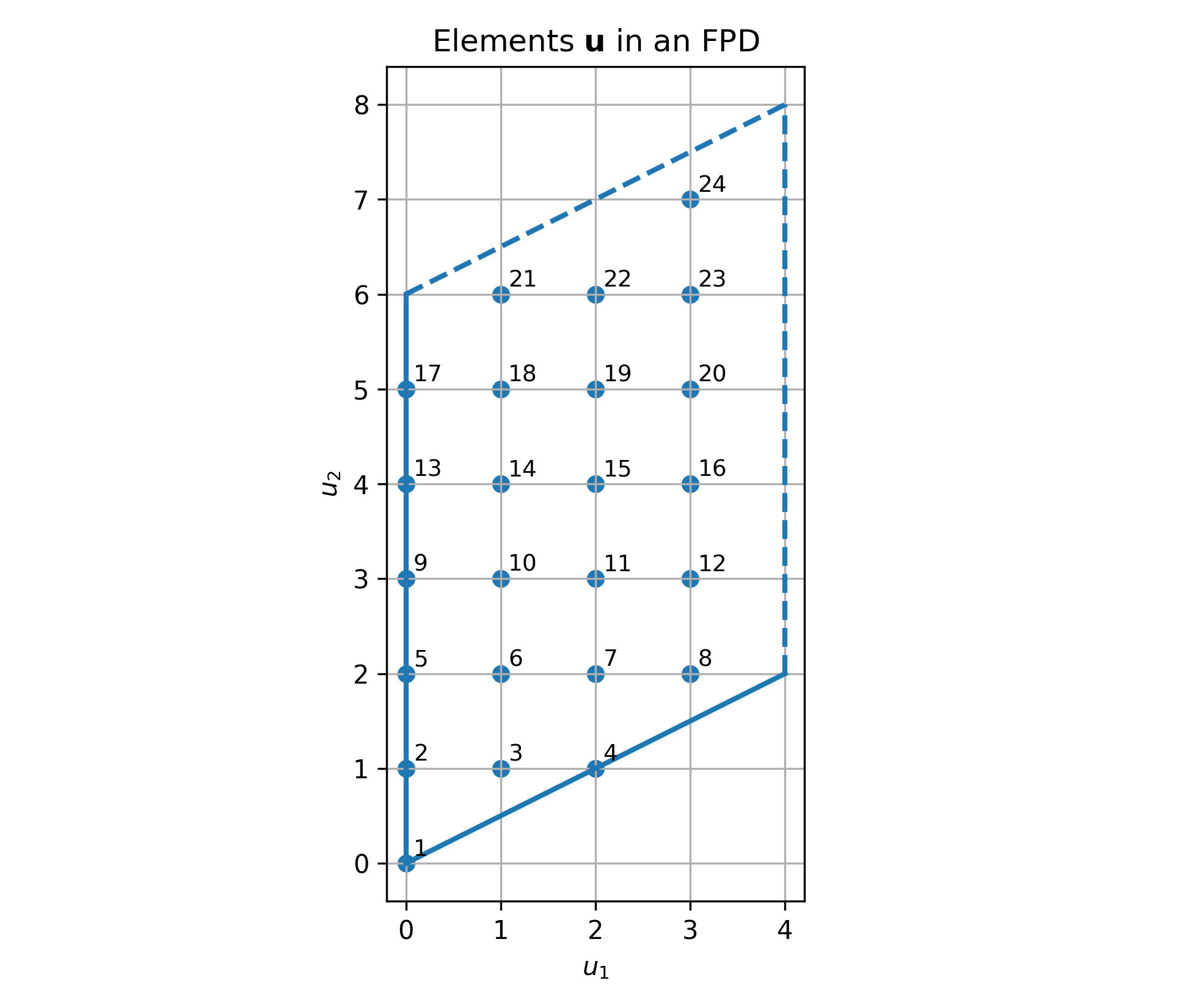}
  \caption{FPD of $\mathbf N^{\top}$.}
  \label{fig:ex:fpd}
\end{figure}

\begin{figure}[htbp]
  \centering
  \includegraphics[width=0.8\linewidth]{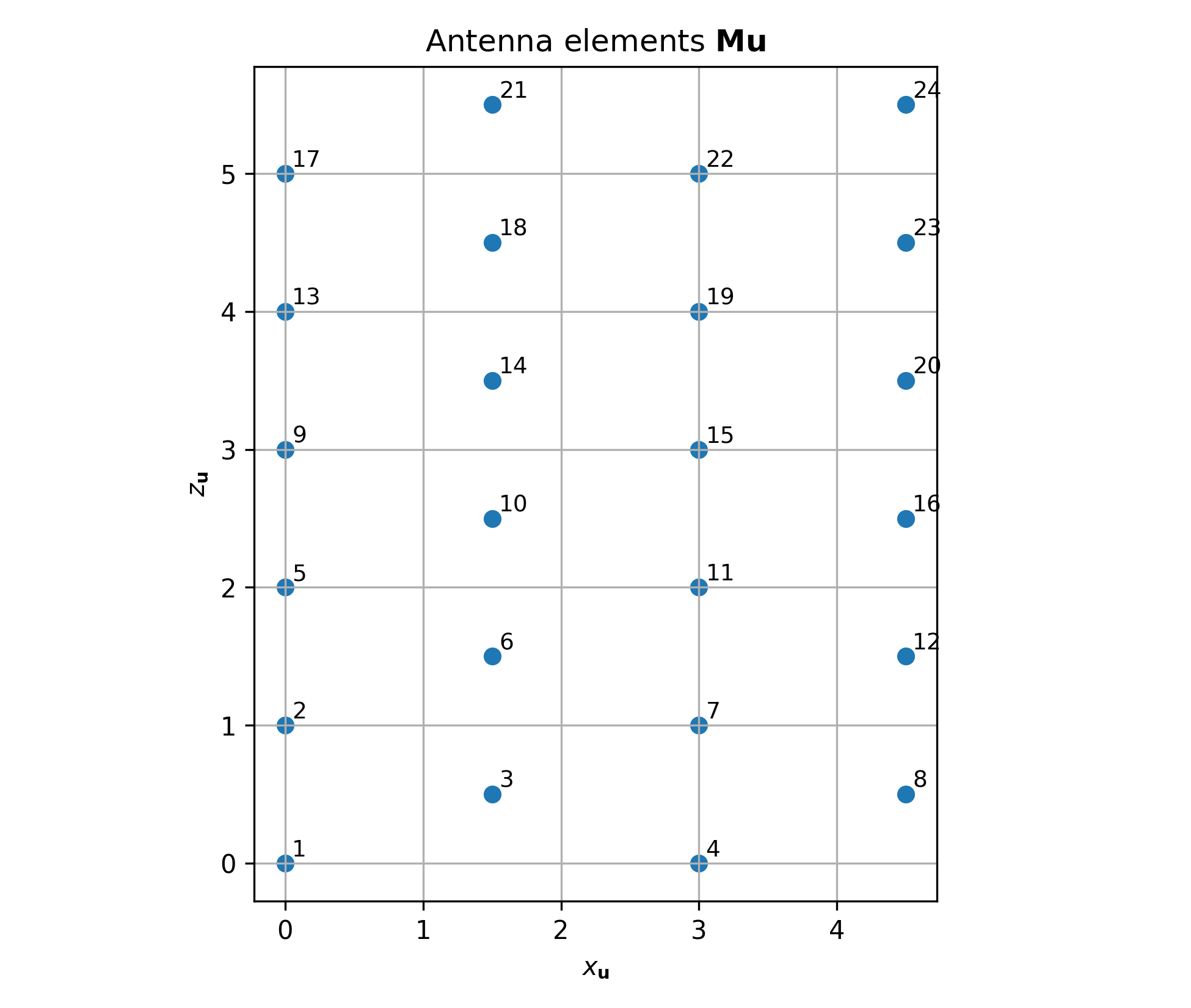}
  \caption{Planar Antenna Array of Example \ref{ex:array}.}
  \label{fig:ex:array}
\end{figure}

Choose the scaling factor $\alpha = 4$.
Then
\begin{equation}\notag
\mathbf N
=
\alpha\,\operatorname{adj}(\mathbf M^{\top})
=
\begin{pmatrix}
4 & 2\\
0 & 6
\end{pmatrix}.
\end{equation}
The FPD of
$\mathbf N^{\top}$ is illustrated in Fig.~\ref{fig:ex:fpd}.
All integer index vectors $\mathbf u$ inside this region are marked in the figure.
In this example, the region contains $|\det(\mathbf N^{\top})|=24$ vectors.

For each index vector $\mathbf u$ in $\mathcal N(\mathbf N^{\top})$
there exists a corresponding antenna location
$\mathbf M\mathbf u$.
The positions of all 24 antenna elements are shown in Fig.~\ref{fig:ex:array}.
\end{example}

This example illustrates how the antenna elements should be arranged so that the antenna indices exactly cover one FPD of $\mathbf N^\top$, and therefore satisfy the requirements of the above 2D-DFT formulation.

\section{Ambiguity Analysis and Robust Recovery}\label{s4}

In this section, we address the ambiguity caused by the vector remainder structure derived in Section~\ref{s3}. We first show that, for a single planar array, the unambiguous range is fundamentally limited, so ambiguity arises once the true parameter vector lies outside this range. We then develop a multi-subarray framework that combines vector remainders with different matrix moduli through the MD-CRT to recover the true parameter vector over an enlarged unambiguous range. Finally, we analyze the effects of practical errors, including quantization and noise, and establish sufficient conditions together with explicit reconstruction error bounds for robust recovery.

\subsection{Multi-Subarray Ambiguity Resolution}
From \eqref{eq:MDDFT_main}, the 2D-DFT yields only the vector remainder $\mathbf r$ of $\mathbf f$ modulo $\mathbf N$, rather than $\mathbf f$ itself. In this subsection, we continue to consider the ideal case where $\mathbf f$ in \eqref{eq:Nandf} is an integer vector. Then there exists an integer vector $\mathbf n\in\mathbb Z^2$ such that
\begin{equation}\label{eq:f_ambiguity}
\mathbf f = \mathbf N \mathbf n + \mathbf r, \quad \mathbf r \in \mathcal N(\mathbf N).
\end{equation}
Hence, a single planar array recovers the exact vector $\mathbf f$ only when $\mathbf f \in \mathcal N(\mathbf N)$. Otherwise, only its folded version $\mathbf r$ is observed. Since $\mathbf f$ depends directly on the motion-induced shift and the target height, when the target velocity or height is large, $\mathbf f$ may fall outside $\mathcal N(\mathbf N)$, and ambiguity then occurs.

To deal with this problem, we employ multiple planar subarrays with different lattice generators similar to the one-dimensional non-uniform linear array framework in~\cite{gangli}.
The key idea is that different subarrays are constructed with different lattice generators, which in turn induce different matrix moduli and hence generate distinct folded observations of $\mathbf g$.
As a result, these folded observations from different subarrays form a system of matrix congruences, which can be jointly solved by the MD-CRT to enlarge the unambiguous range.

Let $\mathbf M_j$, $j=1,2,\dots,J$, be $2\times 2$ lattice generator matrices.
For each $\mathbf M_j$, define the corresponding 2D-DFT matrix modulus
\begin{equation}
\mathbf N_j
\triangleq
\alpha_j \,\operatorname{adj}(\mathbf M_j^{\top}),
\qquad j=1,\dots,J,
\label{eq:Nj_def}
\end{equation}
where $\alpha_j$ is chosen such that $\mathbf N_j$ is an integer matrix for each $j$.
The $j$-th planar subarray places its antenna elements at
$\mathbf M_j \mathbf u$, $\mathbf u \in \mathcal N(\mathbf N_j^{\top})$,
so that the antenna indices cover one whole FPD of $\mathbf N_j^{\top}$.
In this way, we obtain $J$ subarrays, each associated with a different matrix modulus $\mathbf N_j$.

The complex image obtained from the $j$-th subarray is
\[
S_{\mathbf u}^{(j)}(n)
=
C_0^{(j)}
\exp\!\left(
j2\pi\,\mathbf f_j^{\top}
\mathbf N_j^{-\top}\mathbf u
\right)
\delta\!\big(n-n_{x_0}-\Delta_{\rm shift}\big),
\]
where
\[
\mathbf f_j
=
\frac{\alpha_j \det(\mathbf M_j)}{R\lambda}\,\mathbf g,
\]
which, for description convenience, is also assumed to be an integer vector now and the general case will be studied later. By applying the MD-DFT with respect to
$\mathbf u \in \mathcal N(\mathbf N_j^{\top})$ as in
\eqref{eq:Su_final}--\eqref{eq:f_congruence_main},
we obtain
\[
\mathbf f_j
\equiv
\mathbf r_j
\quad
\mathrm{mod}\ \mathbf N_j.
\]
Substituting the expression of $\mathbf f_j$ into the above congruence gives
\begin{equation}
\mathbf g
\equiv
\frac{R\lambda}{\alpha_j \det(\mathbf M_j)}\,\mathbf r_j
\quad \mathrm{mod}\ 
\frac{R\lambda}{\alpha_j \det(\mathbf M_j)}\,\mathbf N_j,
\label{eq:g_multi_arrays}
\end{equation}
for $j=1,2,\cdots,J$.

For notational convenience, define
\begin{equation}
\beta_j
\triangleq
\frac{R\lambda}{\alpha_j \det(\mathbf M_j)},
\qquad j=1,\dots,J.
\label{eq:Rj_def}
\end{equation}
Then \eqref{eq:g_multi_arrays} can be written compactly as
\begin{equation}
\mathbf g
\equiv
\beta_j\,\mathbf r_j
\quad
\mathrm{mod}\ \beta_j\,\mathbf N_j,
\qquad j=1,\dots,J.
\label{eq:g_multi_compact}
\end{equation}
To apply the MD-CRT, the congruence system in \eqref{eq:g_multi_compact} must be represented in integer-valued form. This can be handled in practice by appropriate scaling or quantization. For simplicity, we assume in this subsection that \eqref{eq:g_multi_compact} already satisfies the integer-valued requirement of the MD-CRT. Under this assumption, the resulting integer-valued congruence system can be solved using the MD-CRT.

\begin{prop}[MD-CRT~\cite{MD1}]\label{prop:MD-CRT}
Let $\mathbf{R}_1,\mathbf{R}_2,\ldots,\mathbf{R}_J$ be nonsingular integer matrices of size $D\times D$, and let $\mathbf{R}$ be any lcrm of them.
For any integer vector $\mathbf{g}\in\mathbb{Z}^{D}$, it can be uniquely determined from its vector remainders $\mathbf{g}\equiv\mathbf{r}_j\bmod\mathbf{R}_j$, $1\le j\le J$, if and only if $\mathbf{g}\in\mathcal{N}(\mathbf{R})$.
\end{prop}

Under the above integer-valued assumption, applying Proposition~\ref{prop:MD-CRT} to the system in \eqref{eq:g_multi_compact}, the parameter vector $\mathbf g$ can be uniquely recovered within the FPD of
\[
\mathbf R \triangleq \operatorname{lcrm}\!\big(\beta_1\mathbf N_1,\dots,\beta_J\mathbf N_J\big).
\]
Therefore, the FPD $\mathcal N(\mathbf R)$ defines an \textit{unambiguous range} of the multi-subarray system. Since the lcrm is not unique in general, the corresponding unambiguous range is not unique either. However, the absolute determinant value $|\det(\mathbf R)|$ is uniquely determined by $\beta_1\mathbf N_1,\dots,\beta_J\mathbf N_J$ for any lcrm of them. Since $|\mathcal N(\mathbf R)|=|\det(\mathbf R)|$, the size of an unambiguous range, i.e., the number of uniquely recoverable unknown integer vectors $\mathbf{g}$, is uniquely determined, although its shape may depend on the particular choice of the lcrm. In practice, one usually fixes a particular lcrm in advance and uses the corresponding FPD as the unambiguous range.

If there does not exist a matrix among $\beta_1\mathbf N_1,\dots,\beta_J\mathbf N_J$ that is a right multiple of all the others, then the FPD of the lcrm $\mathbf R$ contains strictly more points than that of any individual matrix modulus.
That is,
\[
|\det(\mathbf R)| > |\det(\beta_j\mathbf N_j)|,
\qquad j=1,\dots,J.
\]
Therefore, the size of an unambiguous range achieved by multiple subarrays is strictly larger than that provided by any single subarray.

To illustrate this enlargement geometrically, we consider the following example.

\begin{example}\label{ex:lcrm_fpd}
Consider two matrices
\begin{equation}\notag
\mathbf R_1
=
\begin{pmatrix}
2 & 1\\
0 & 2
\end{pmatrix},
\qquad
\mathbf R_2
=
\begin{pmatrix}
2 & 0\\
1 & 2
\end{pmatrix}.
\end{equation}
It can be verified that their least common right multiple is \cite{guo25}
\begin{equation}\notag
\mathbf R
=
\operatorname{lcrm}(\mathbf R_1,\mathbf R_2)
=
\begin{pmatrix}
4 & 0\\
0 & 4
\end{pmatrix}.
\end{equation}
The corresponding FPDs are illustrated in Fig.~\ref{fig:ex:2}. As shown in the figure, the FPD of $\mathbf R$ strictly contains those of $\mathbf R_1$ and $\mathbf R_2$. Therefore, the unambiguous range achieved by the two planar subarrays generated by $\mathbf{R}_1$ and $\mathbf{R}_2$ is strictly larger than that of any single subarray alone.

\begin{figure}[htbp]
  \centering
  \includegraphics[width=0.7\linewidth]{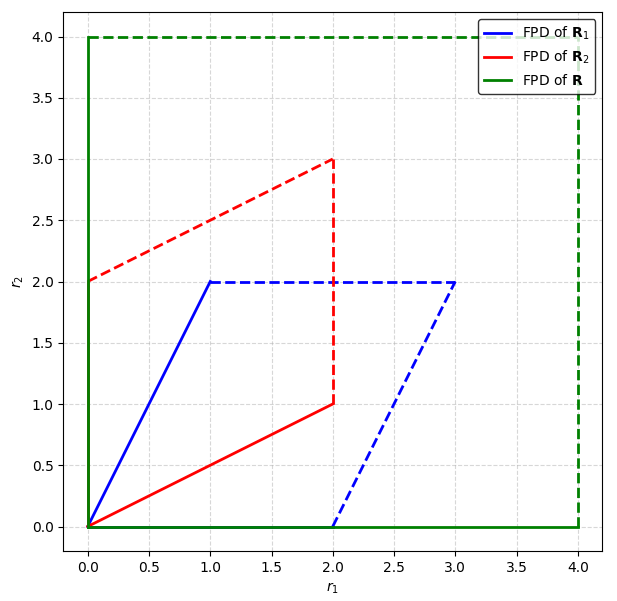}
  \caption{The FPDs of Example \ref{ex:lcrm_fpd}.}
  \label{fig:ex:2}
\end{figure}
\end{example}

This example shows geometrically how the lcrm can generate an FPD that is strictly larger than any associated with an individual matrix modulus, and therefore how multiple subarrays enlarge the unambiguous range.

\subsection{Approximate 2D-DFT Peak Model for Non-Integer Frequency Vectors}\label{s4.2}

In general, the frequency vector $\mathbf f_j$ is not necessarily integer-valued. 
This makes it harder to interpret the detected 2D-DFT peak. 
When $\mathbf f_j\in\mathbb Z^2$, the detected peak can be identified directly with the standard vector remainder of $\mathbf f_j$ modulo $\mathbf N_j$. 
However, when $\mathbf f_j\notin\mathbb Z^2$, this interpretation may not be exact. 
Specifically, let
\[
\mathbf N_j=\mathbf U_j\mathbf D_j\mathbf V_j
\]
be a Smith normal form~\cite{matrix} of $\mathbf N_j$, where $\mathbf U_j,\mathbf V_j$ are unimodular and $\mathbf D_j$ is diagonal. Then the exact peak selection rule is governed by the transformed frequency vector
\begin{equation}\label{eq:transformed_f_j}
\widetilde{\mathbf f}_j=\mathbf U_j^{-1}\mathbf f_j
\end{equation}
on the diagonal grid induced by $\mathbf D_j$. 
Equivalently, in the original coordinates, the exact decision regions are generally skewed and depend on the unimodular matrix $\mathbf U_j$. 
As a result, a direct closed-form characterization of the exact peak is difficult, and the subsequent error analysis would have to track these transformed decision boundaries explicitly.

For this reason, we adopt a tractable approximation that replaces the non-integer vector $\mathbf f_j$ by a nearby integer vector. 
Define the componentwise nearest-integer operator $[\cdot]:\mathbb R^2\to\mathbb Z^2$ by
\begin{equation}\label{eq:nearest_integer_operator}
[\mathbf a]
\triangleq
\big[[a_1]_{\mathbb Z},[a_2]_{\mathbb Z}\big]^\top,
\qquad
[a_i]_{\mathbb Z}\in\arg\min_{n\in\mathbb Z}|a_i-n|,
\end{equation}
for $\mathbf a=[a_1,a_2]^\top\in\mathbb R^2$. 
When $a_i\notin \mathbb Z+\tfrac12$, the minimizer is unique. When $a_i\in \mathbb Z+\tfrac12$, either of the two nearest integers may be chosen arbitrarily.
We then define the integer approximation of $\mathbf f_j$ as
\begin{equation}\label{eq:f_int_definition_paper}
\mathbf f_j^{\mathrm{int}}
\triangleq
[\mathbf f_j],
\end{equation}
and the associated quantization error as
\begin{equation}\label{eq:quantization_error_definition_paper}
\mathbf e_j
\triangleq
\mathbf f_j^{\mathrm{int}}-\mathbf f_j.
\end{equation}
Then,
\begin{equation}\label{eq:quantization_error_bound_paper}
\|\mathbf e_j\|_\infty \leq \frac12.
\end{equation}
Using this integer approximation, we approximate the detected 2D-DFT peak by the standard vector remainder of $\mathbf f_j^{\mathrm{int}}$ modulo $\mathbf N_j$. 
That is, we use the vector remainder of $\mathbf f_j^{\mathrm{int}}$ modulo $\mathbf N_j$ as an approximation of the detected peak. 
This converts the problem with non-integer frequency vectors into the same integer vector remainder framework used in the previous subsection, so that the subsequent MD-CRT reconstruction analysis can still be applied. 

This approximation is reasonable for the following reason. 
Substituting \eqref{eq:quantization_error_definition_paper} into \eqref{eq:transformed_f_j}, we have
\[
\widetilde{\mathbf f}_j
=
\mathbf U_j^{-1}\mathbf f_j
=
\mathbf U_j^{-1}\mathbf f_j^{\mathrm{int}}
-
\mathbf U_j^{-1}\mathbf e_j.
\]
Since $\mathbf U_j^{-1}\mathbf f_j^{\mathrm{int}}\in\mathbb Z^2$, the transformed vector $\widetilde{\mathbf f}_j$ is obtained by adding the error vector $-\mathbf U_j^{-1}\mathbf e_j$ to an integer lattice point. 
If 
\begin{equation}\label{eq:no_boundary_crossing_condition_paper}
\|\mathbf U_j^{-1}\mathbf e_j\|_\infty \leq \frac12,
\end{equation}
then $\mathbf U_j^{-1}\mathbf f_j^{\mathrm{int}}$ is the nearest integer vector to $\widetilde{\mathbf f}_j$. 
Therefore, in the transformed coordinates, the detected 2D-DFT peak corresponds exactly to the integer grid point $\mathbf U_j^{-1}\mathbf f_j^{\mathrm{int}}$. Equivalently, in the original coordinates, this means that the detected peak can be represented by the vector remainder of $\mathbf f_j^{\mathrm{int}}$ modulo $\mathbf N_j$.
Hence, under \eqref{eq:no_boundary_crossing_condition_paper}, the approximation is exact. 
Accordingly, in the following analysis, we adopt this approximation and represent the detected 2D-DFT peak in the non-integer case by the vector remainder of $\mathbf f_j^{\mathrm{int}}$ modulo $\mathbf N_j$. 
A more precise analysis of the exact peak selection rule in the non-integer vector case is left for future work.

\subsection{Accuracy and Robustness Analysis}\label{s4.3}

In the previous two subsections, we established the multi-subarray framework to enlarge the unambiguous range and introduced an approximate 2D-DFT peak model for non-integer frequency vectors $\mathbf{f}$. When the matrix moduli used in the MD-CRT are fixed, replacing the MD-CRT by its robust version does not change the size of the unambiguous range, because the corresponding lcrm remains the same. Therefore, the main issue in this subsection is not the unambiguous range itself, but the sensitivity of the recovery result to practical errors in the detected vector remainders. We now analyze how these errors affect the recovery of $\mathbf g$, and derive conditions for robust reconstruction.

In the non-integer case, intrinsic vector remainder errors arise even without additive noise. These errors come from the quantization of the frequency vectors and the subsequent rounding of the scaled parameter vector. We first analyze this noiseless case, and then extend the result to noisy vector remainder detection.

\subsubsection{Quantization Errors for Frequency Vector $\mathbf{f}$}

We first analyze the noiseless case. Recall that, for the $j$-th subarray,
\[
\mathbf f_j=\frac{1}{\beta_j}\mathbf g,
\qquad j=1,\dots,J.
\]
As established in Subsection~\ref{s4.2}, in the non-integer case we represent the detected 2D-DFT peak by the vector remainder of
\[
\mathbf f_j^{\mathrm{int}}=[\mathbf f_j],
\]
and define the corresponding quantization error by
\begin{equation}\label{eq:ej_def}
\mathbf e_j
\triangleq
\mathbf f_j^{\mathrm{int}}-\mathbf f_j.
\end{equation}
Then,
\begin{equation}\label{eq:ej_bound_inf}
\|\mathbf e_j\|_\infty \leq \frac{1}{2},
\qquad j=1,\dots,J.
\end{equation}
Therefore, there exists an integer vector $\mathbf n_j\in\mathbb Z^2$ such that
\begin{equation}\label{eq:fint_rem_rep}
\mathbf f_j^{\mathrm{int}}
=
\mathbf N_j\mathbf n_j+\mathbf r_j,
\qquad
\mathbf r_j\in\mathcal N(\mathbf N_j),
\end{equation}
where $\mathbf r_j$ is the vector remainder of $\mathbf f_j^{\mathrm{int}}$ modulo $\mathbf N_j$.
Substituting $\mathbf f_j^{\mathrm{int}}=\mathbf f_j+\mathbf e_j$ and $\mathbf f_j=\mathbf g/\beta_j$ into \eqref{eq:fint_rem_rep} yields
\begin{equation}\label{eq:g_beta_e_relation}
\mathbf g+\beta_j\mathbf e_j
=
\beta_j\mathbf N_j\mathbf n_j+\beta_j\mathbf r_j,
\qquad j=1,\dots,J.
\end{equation}

We next convert \eqref{eq:g_beta_e_relation} into an integer-valued congruence system suitable for robust MD-CRT. Specifically, introduce a common positive integer $k$ such that the components of
\begin{equation}\label{eq:k_clear}
\mathbf A_j
\triangleq
k\beta_j\mathbf N_j \in \mathbb Z^{2\times 2},
\qquad
\mathbf b_j
\triangleq
k\beta_j\mathbf r_j \in \mathbb Z^2,
\end{equation}
are all integers, for all $j=1,\dots,J$.
Multiplying \eqref{eq:g_beta_e_relation} by $k$ gives
\begin{equation}\label{eq:kg_relation}
k\mathbf g + k\beta_j\mathbf e_j
=
\mathbf A_j\mathbf n_j+\mathbf b_j,
\qquad j=1,\dots,J.
\end{equation}
The right hand side of \eqref{eq:kg_relation} is integer-valued, while $k\mathbf g$ is not necessarily an integer vector. Define
\begin{equation}\label{eq:delta_def}
\mathbf x
\triangleq
[k\mathbf g]
\in\mathbb Z^2 \quad
\text{and} \quad
\boldsymbol\delta
\triangleq
\mathbf x-k\mathbf g.
\end{equation}
Then
\begin{equation}\label{eq:delta_bound}
\|\boldsymbol\delta\|_\infty \le \frac{1}{2}.
\end{equation}
Substituting \eqref{eq:delta_def} into \eqref{eq:kg_relation}, we obtain
\begin{equation}\label{eq:x_congruence_with_error}
\mathbf x
=
\mathbf A_j\mathbf n_j+\mathbf b_j+\boldsymbol\eta_j,
\qquad j=1,\dots,J,
\end{equation}
where
\begin{equation}\label{eq:eta_def}
\boldsymbol\eta_j
\triangleq
\boldsymbol\delta-k\beta_j\mathbf e_j.
\end{equation}
Using \eqref{eq:ej_bound_inf} and \eqref{eq:delta_bound}, we have
\begin{equation}\label{eq:eta_bound}
\|\boldsymbol\eta_j\|_\infty
\le
\|\boldsymbol\delta\|_\infty+\|k\beta_j\mathbf e_j\|_\infty
\le
\frac{1}{2}+\frac{k|\beta_j|}{2}.
\end{equation}

Equation \eqref{eq:x_congruence_with_error} reformulates the recovery problem as an integer-valued congruence system for the unknown integer vector $\mathbf x=[k\mathbf g]$. Once $\mathbf x$ is estimated, the original parameter vector $\mathbf g$ is recovered by dividing it by $k$. Here, the detected quantities $\mathbf b_j=k\beta_j\mathbf r_j$ are no longer exact vector remainders of $\mathbf x$, because they are accompanied by the errors $\boldsymbol\eta_j$ caused by quantization. Since the MD-CRT is sensitive to such vector remainder errors, robust MD-CRT is required for robust recovery.

To state the robust recovery condition, we next introduce the quantity that determines the allowable error threshold in the robust MD-CRT \cite{MD2,mstage-mdcrt}. Let $\mathbf G_{i,j}$ be a greatest common left divisor (gcld) of $\mathbf A_i$ and $\mathbf A_j$, and define
\begin{equation}\label{eq:Lambda_def}
\Lambda
\triangleq
\max_{1\le i\le J}
\min_{\substack{1\le j\le J\\ j\neq i}}
\lambda_{\mathcal L(\mathbf G_{i,j})},
\end{equation}
where $\lambda_{\mathcal L(\mathbf G_{i,j})}$ denotes the minimum nonzero distance of the lattice generated by $\mathbf G_{i,j}$ under the $\ell_\infty$ norm. 
With this notation, the erroneous system in \eqref{eq:x_congruence_with_error} is now in the standard form required by robust MD-CRT.
We can therefore apply the corresponding result to characterize when the integer vector $\mathbf x$ can be robustly recovered from the matrix moduli $\{\mathbf A_j\}_{j=1}^J$ and the detected vector remainders $\{\mathbf b_j\}_{j=1}^J$.

\begin{prop}[Robust MD-CRT~\cite{MD2}]\label{prop:robust_mdcrt_specialized}
Consider the erroneous congruence system in \eqref{eq:x_congruence_with_error}.
If there exists a constant $\tau_x$ such that
\begin{equation}\label{eq:tau_x_condition}
\|\boldsymbol\eta_j\|_\infty \le \tau_x < \frac{\Lambda}{4},
\qquad j=1,\dots,J,
\end{equation}
then the integer vector $\mathbf x$ can be robustly recovered from $\{\mathbf A_j\}_{j=1}^J$ and $\{\mathbf b_j\}_{j=1}^J$ and the resulting estimate $\hat{\mathbf x}$ satisfies
\begin{equation}\label{eq:x_error_bound}
\|\hat{\mathbf x}-\mathbf x\|_\infty \le \tau_x.
\end{equation}
\end{prop}

This proposition directly yields the following recovery bounds for the parameter vector $\mathbf g$.
The corresponding estimate of $\mathbf g$ is given by
\begin{equation}\label{eq:g_hat_from_x}
\hat{\mathbf g}
\triangleq
\frac{1}{k}\hat{\mathbf x}.
\end{equation}
Since $\mathbf x=k\mathbf g+\boldsymbol\delta$, we have
\[
\hat{\mathbf g}-\mathbf g
=
\frac{1}{k}(\hat{\mathbf x}-\mathbf x)
+
\frac{1}{k}\boldsymbol\delta.
\]
Using \eqref{eq:delta_bound} and \eqref{eq:x_error_bound}, we obtain
\begin{equation}\label{eq:g_total_bound}
\|\hat{\mathbf g}-\mathbf g\|_\infty
\le
\frac{\tau_x}{k}+\frac{1}{2k}.
\end{equation}

Let
\begin{equation}\label{eq:beta_max_def}
\beta_{\max}
\triangleq
\max_{1\le j\le J}|\beta_j|.
\end{equation}
From \eqref{eq:eta_bound} and \eqref{eq:tau_x_condition}, a sufficient condition for robust recovery in the noiseless case is
\begin{equation}\label{eq:robust_sufficient}
\frac{1}{2}+\frac{k\beta_{\max}}{2}
<
\frac{\Lambda}{4}.
\end{equation}
Under \eqref{eq:robust_sufficient}, we may choose
\begin{equation}\label{eq:tau_x_choice}
\tau_x
\triangleq
\frac{1}{2}+\frac{k\beta_{\max}}{2},
\end{equation}
which satisfies \eqref{eq:tau_x_condition}.
Substituting \eqref{eq:tau_x_choice} into \eqref{eq:g_total_bound} gives the reconstruction error bound
\begin{equation}\label{eq:g_error_explicit}
\|\hat{\mathbf g}-\mathbf g\|_\infty
\le
\frac{\beta_{\max}}{2}+\frac{1}{k}.
\end{equation}
This bound characterizes the recovery accuracy of $\mathbf g$ in the noiseless case.

We summarize the above derivation by the following noiseless robust recovery result. For comparison, related robust recovery results for one-dimensional linear arrays with two subarrays under quantization errors were studied in~\cite{gangli}.

\begin{thm}[Noiseless robust recovery]\label{thm:noiseless_recovery}
If the sufficient condition in \eqref{eq:robust_sufficient} holds, then the parameter vector $\mathbf g$ can be robustly recovered.
Moreover, the reconstruction error is bounded as in \eqref{eq:g_error_explicit}.
\end{thm}

\subsubsection{Remainder Errors Induced from Noise}
We now consider additive noises in signals that may cause errors in the detected vector remainders, in addition to the intrinsic quantization errors already analyzed above.
Suppose that the detected vector remainder for the $j$-th subarray is
\[
\tilde{\mathbf r}_j
=
\mathbf r_j+\Delta\mathbf r_j,
\qquad j=1,\dots,J,
\]
where $\Delta\mathbf r_j$ denotes the vector remainder error caused by noise.

Recall from \eqref{eq:k_clear} that, in the noiseless case, we multiply each vector remainder $\beta_j \mathbf r_j$ by the factor $k$ so that the resulting quantity
$\mathbf b_j = k\beta_j\mathbf r_j$
is integer-valued. 
In the noisy case, using the same scaling as in \eqref{eq:k_clear}, we obtain the corresponding observed integer vector remainder
\begin{equation}\label{eq:b_tilde_def}
\tilde{\mathbf b}_j
\triangleq
k\beta_j\tilde{\mathbf r}_j
=
\mathbf b_j+k\beta_j\Delta\mathbf r_j.
\end{equation}
The derivation now proceeds in exactly the same way as in the noiseless case. Repeating the above steps, we obtain
\begin{equation}\label{eq:x_congruence_with_noise}
\mathbf x
=
\mathbf A_j\mathbf n_j+\tilde{\mathbf b}_j+\tilde{\boldsymbol\eta}_j,
\qquad j=1,\dots,J,
\end{equation}
where the resulting vector remainder error is
\begin{equation}\label{eq:eta_noise_def}
\tilde{\boldsymbol\eta}_j
=
\boldsymbol\delta-k\beta_j\mathbf e_j-k\beta_j\Delta\mathbf r_j.
\end{equation}
Hence,
\begin{equation}\label{eq:eta_noise_bound}
\|\tilde{\boldsymbol\eta}_j\|_\infty
\le
\frac{1}{2}
+
\frac{k|\beta_j|}{2}
+
k|\beta_j|\,\|\Delta\mathbf r_j\|_\infty.
\end{equation}

Therefore, the noisy case in \eqref{eq:x_congruence_with_noise} remains within the robust MD-CRT framework.
If there exists $\tau_x$ such that
\begin{equation}\label{eq:tau_x_condition_noise}
\|\tilde{\boldsymbol\eta}_j\|_\infty \le \tau_x < \frac{\Lambda}{4},
\qquad j=1,\dots,J,
\end{equation}
then Proposition~\ref{prop:robust_mdcrt_specialized} still applies, and the resulting estimate $\hat{\mathbf x}$ satisfies
\[
\|\hat{\mathbf x}-\mathbf x\|_\infty \le \tau_x.
\]
Consequently,
\begin{equation}\label{eq:g_error_noise_general}
\|\hat{\mathbf g}-\mathbf g\|_\infty
\le
\frac{\tau_x}{k}+\frac{1}{2k}.
\end{equation}
Using \eqref{eq:eta_noise_bound} and \eqref{eq:tau_x_condition_noise}, a sufficient condition for robust recovery in the noisy case is
\begin{equation}\label{eq:robust_sufficient_noise}
\frac{1}{2}
+
\frac{k\beta_{\max}}{2}
+
k\max_{1\le j\le J}\big(|\beta_j|\,\|\Delta\mathbf r_j\|_\infty\big)
<
\frac{\Lambda}{4}.
\end{equation}
Under \eqref{eq:robust_sufficient_noise}, we may choose
\begin{equation}\label{eq:tau_x_choice_noise}
\tau_x
=
\frac{1}{2}
+
\frac{k\beta_{\max}}{2}
+
k\max_{1\le j\le J}\big(|\beta_j|\,\|\Delta\mathbf r_j\|_\infty\big),
\end{equation}
which satisfies \eqref{eq:tau_x_condition_noise}.
Substituting \eqref{eq:tau_x_choice_noise} into \eqref{eq:g_error_noise_general} gives
\begin{equation}\label{eq:g_noise_final}
\|\hat{\mathbf g}-\mathbf g\|_\infty
\le
\frac{\beta_{\max}}{2}
+
\max_{1\le j\le J}\big(|\beta_j|\,\|\Delta\mathbf r_j\|_\infty\big)
+
\frac{1}{k}.
\end{equation}

We summarize the above analysis by the following noisy robust recovery result.

\begin{thm}[Noisy robust recovery]\label{thm:noisy_recovery}
If the sufficient condition in \eqref{eq:robust_sufficient_noise} holds, then the parameter vector $\mathbf g$ can be robustly recovered under noisy vector remainder detection.
Moreover, the reconstruction error is bounded as in \eqref{eq:g_noise_final}.
\end{thm}

Theorems~\ref{thm:noiseless_recovery} and \ref{thm:noisy_recovery} complete the theoretical framework of this section. Together with the ambiguity analysis in the first subsection, we show that the proposed multi-subarray framework not only enlarges the unambiguous range, but also enables robust recovery of $\mathbf g$ in the presence of practical errors, with explicit reconstruction error bounds determined by quantization and remainder error levels.

Moreover, the bounds in \eqref{eq:g_error_explicit} and \eqref{eq:g_noise_final} directly determine the recovery accuracy of $\mathbf g$. 
In practice, we usually want this bound to be as small as possible. 
From \eqref{eq:g_noise_final}, one can see that the recovery accuracy is closely related to the quantities $\{\beta_j\}_{j=1}^J$. 
Since $\beta_j=\frac{R\lambda}{\alpha_j\det(\mathbf M_j)}$,
increasing $\alpha_j$ decreases $\beta_j$, and hence reduces both the intrinsic quantization term $\beta_{\max}/2$ and the noise related term
$\max_{1\le j\le J}\big(|\beta_j|\,\|\Delta\mathbf r_j\|_\infty\big)$.
Therefore, a larger $\alpha_j$ leads to a smaller reconstruction error bound and hence a more accurate recovery of $\mathbf g$.
However, as mentioned earlier, a larger $\alpha_j$ leads to a larger size of array, which may not be possible for a limited platform size. 
When the physical platform size is fixed, this improvement can be, however, achieved through zero-padding in the 2D-DFT.

To see this, for a fixed array geometry $\mathbf{M}$, consider two positive scaling factors $0<\alpha^{(1)}<\alpha^{(2)}$ such that
$\mathbf N^{(1)}=\alpha^{(1)}\operatorname{adj}(\mathbf M^\top)$ and
$\mathbf N^{(2)}=\alpha^{(2)}\operatorname{adj}(\mathbf M^\top)$ are both integer matrices. The smaller factor $\alpha^{(1)}$ corresponds to the actual planar array, in the sense that antenna measurements are available for all indices
$\mathbf u\in\mathcal N\!\big((\mathbf N^{(1)})^\top\big)$.
Since $\alpha^{(2)}>\alpha^{(1)}$, the corresponding larger index set $\mathcal N\!\big((\mathbf N^{(2)})^\top\big)$ contains $\mathcal N\!\big((\mathbf N^{(1)})^\top\big)$ as a subset. Hence, to perform the 2D-DFT with the size associated with $\alpha^{(2)}$, we keep the measured data on $\mathcal N\!\big((\mathbf N^{(1)})^\top\big)$ and assign zero values to those indices in $\mathcal N\!\big((\mathbf N^{(2)})^\top\big)\setminus \mathcal N\!\big((\mathbf N^{(1)})^\top\big)$, i.e., zero-padding in two-dimensions. In this way, zero-padding provides a practical way to realize a larger size 2D-DFT with a larger effective $\alpha$, and hence to improve the reconstruction accuracy without changing the physical array. This follows the same idea of zero-padding in the one-dimensional case in~\cite{gangli} to improve the recovery accuracy.

\section{Simulation Results}\label{s5}

In this section, we provide a representative numerical example to illustrate the enlargement of the unambiguous range achieved by the proposed multi-subarray framework. 
We consider only the noiseless case here. 
Since this Part I of the two-part paper only focuses on the basic formulation and recovery framework, we leave detailed planar subarray designs and detailed simulations in Part II of this two-part paper, \cite{part2}. 
In this section we show, through a simple numerical example, that using two planar subarrays jointly can enlarge the unambiguous range compared with using only a single subarray.

\subsection{Simulation Setup}

In the simulation, the radar altitude is set to $H=4000~\mathrm{m}$, the platform velocity is set to $v=200~\mathrm{m/s}$, the cross-range resolution is set to $\Delta_x=2~\mathrm{m}$, and the wavelength is set to $\lambda=0.05~\mathrm{m}$. 
The terrain and road are generated numerically. 
A point target is selected on the road, and its velocity is chosen along the local road direction. 
To emphasize the role of the unambiguous range, we use the same scene setting throughout this section and compare the corresponding reconstructed target locations in three-dimensional space. 
The target position and velocity are chosen as
$\mathbf r_0=[x_0,y_0,z_0]^\top=[175.409811,\ 263.516684,\ 94.868794]^\top~\mathrm{m}$
and
$\mathbf v_0=[v_x,v_y,v_z]^\top=[3.246923,\ 9.450347,\ 8.315553]^\top~\mathrm{m/s}$.

\subsection{Unambiguous Range Enlargement by Two Planar Subarrays}

We consider the proposed co-prime planar array design with two subarrays. 
Let the two lattice generator matrices be
\[
\mathbf M_1 =
\begin{pmatrix}
2 & -1\\
0 & 2
\end{pmatrix},
\qquad
\mathbf M_2 =
\begin{pmatrix}
2 & 0\\
-1 & 2
\end{pmatrix},
\]
and choose the factor $\alpha_1=\alpha_2=5$. 
Then, according to \eqref{eq:Nj_def}, the corresponding integer matrix moduli are
\[
\mathbf N_1 =
\begin{pmatrix}
10 & 0\\
5 & 10
\end{pmatrix},
\qquad
\mathbf N_2 =
\begin{pmatrix}
10 & 5\\
0 & 10
\end{pmatrix}.
\]
The resulting array is shown in Fig.~\ref{fig:s5_coprime_array}. 
The entire array occupies a $19\times 19$ square region.

\begin{figure}[htbp]
    \centering
    \includegraphics[width=0.6\linewidth]{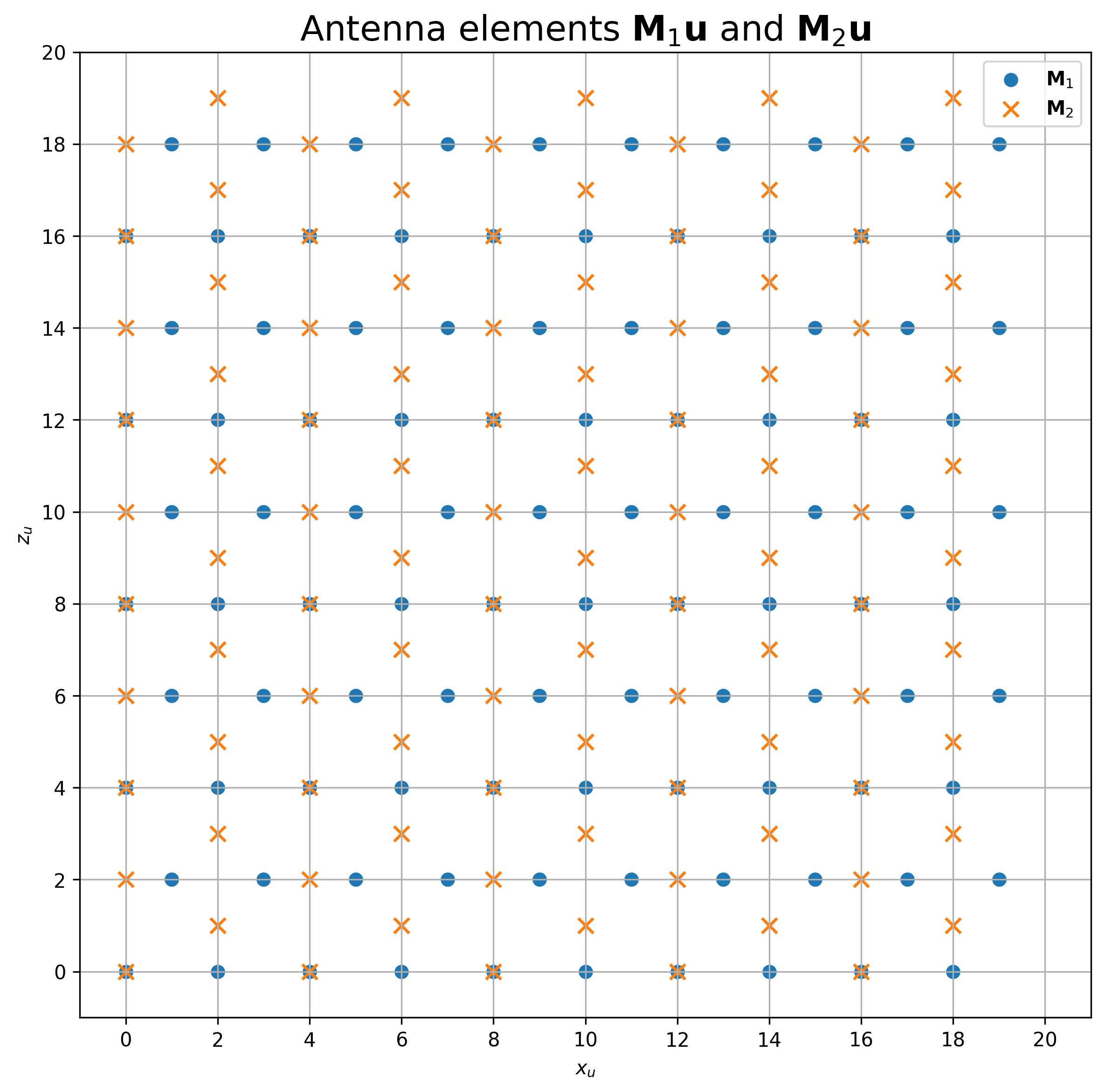}
    \caption{Placement of two co-prime planar subarrays.}
    \label{fig:s5_coprime_array}
\end{figure}

To compare the unambiguous ranges more clearly, we describe them in the discrete setting, namely, in terms of the integer frequency vector and the associated FPD set $\mathcal N(\cdot)$. 
Equivalently, we assume in this comparison that the corresponding integer-valued frequency representation is used, so that the range can be compared directly through set inclusion.

Since
\[
\beta_1=\beta_2=\beta=\frac{R\lambda}{\alpha_j\det(\mathbf M_j)}=\frac{R\lambda}{20},
\]
we have
\[
\operatorname{lcrm}(\beta\mathbf N_1,\beta\mathbf N_2)
=
\beta\,\operatorname{lcrm}(\mathbf N_1,\mathbf N_2)
=
\beta(20\mathbf I)
=
R\lambda\,\mathbf I.
\]
Hence, when both subarrays are used jointly, the corresponding unambiguous range of the parameter vector $\mathbf g$ is
$\mathcal N(R\lambda\,\mathbf I)$.

We next consider the case where only the first subarray generated by $\mathbf M_1$ is used. 
Then there is only one matrix modulus $\mathbf N_1$, and the unambiguous range reduces to
$\mathcal N(\beta\mathbf N_1)$.
Since
\[
\beta\mathbf N_1
=
\frac{R\lambda}{20}
\begin{pmatrix}
10 & 0\\
5 & 10
\end{pmatrix}
=
R\lambda
\begin{pmatrix}
\frac12 & 0\\[2mm]
\frac14 & \frac12
\end{pmatrix},
\]
the corresponding unambiguous range is
\[
\mathcal N(\beta\mathbf N_1)
=
\mathcal N\!\left(
R\lambda
\begin{pmatrix}
\frac12 & 0\\[2mm]
\frac14 & \frac12
\end{pmatrix}
\right).
\]
and it can be verified that
\[
\mathcal N(\beta\mathbf N_1)\subsetneq \mathcal N(\beta\mathbf R)=\mathcal N(R\lambda\mathbf I).
\]
This shows that using the two subarrays jointly yields a strictly larger unambiguous range than using only the subarray generated by $\mathbf M_1$.

To further illustrate the effect of unambiguous range enlargement, we show the reconstructed target locations in three-dimensional space for these two cases. 
The reconstructed target positions are marked by red points in the figures.

\begin{figure}[htbp]
    \centering
    \includegraphics[width=0.65\linewidth]{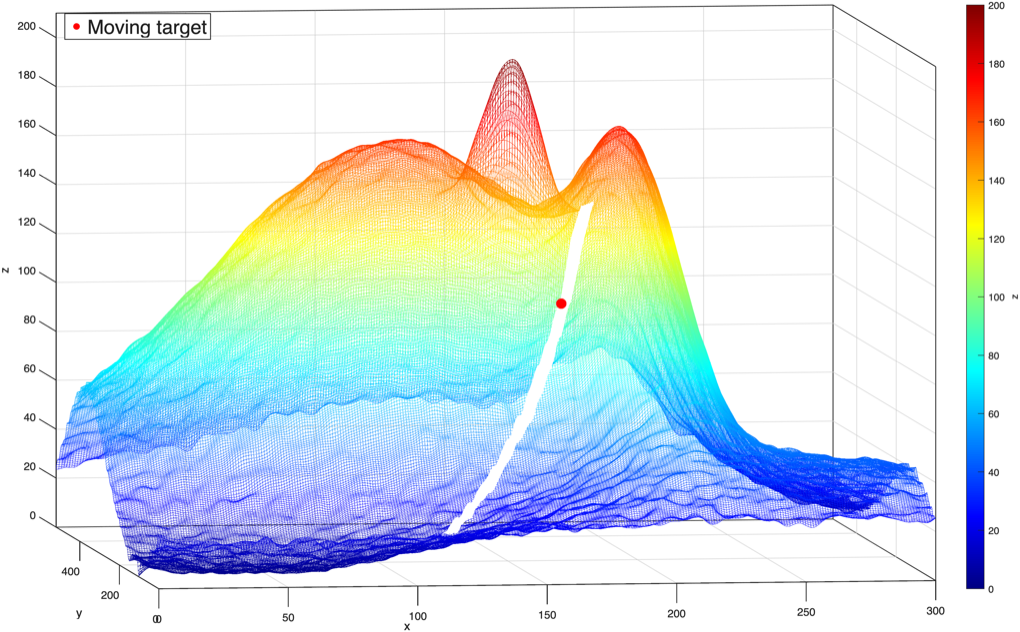}
    \caption{Reconstructed target location for the proposed two-subarray co-prime planar array design. The recovered point remains on the road.}
    \label{fig:s5_design1_3d}
\end{figure}

\begin{figure}[htbp]
    \centering
    \includegraphics[width=0.65\linewidth]{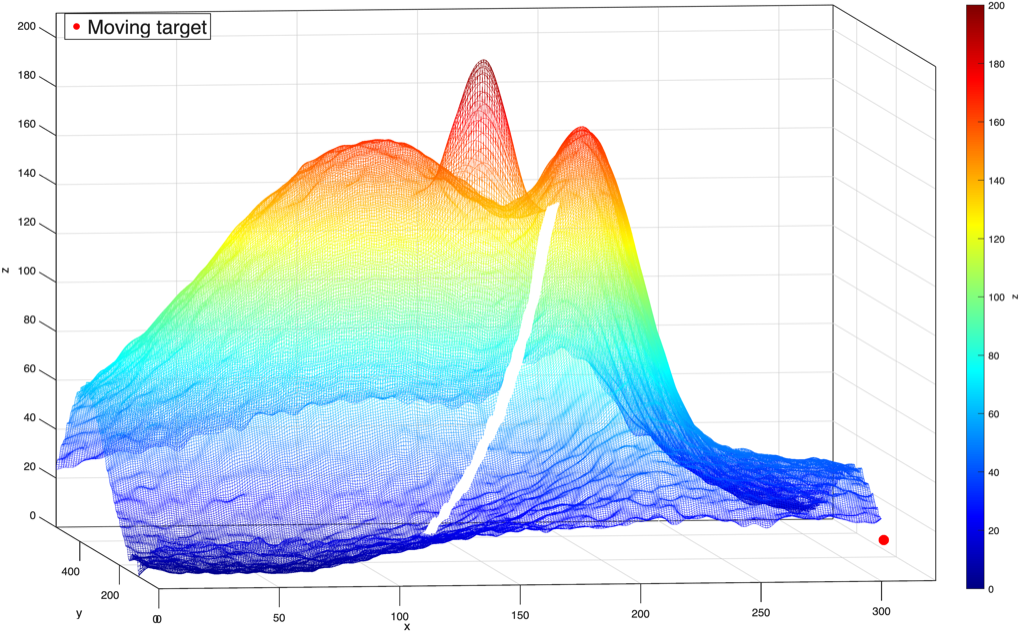}
    \caption{Reconstructed target location using only the subarray generated by $\mathbf M_1$. The recovered point is shifted away from the road.}
    \label{fig:s5_design1_m1_3d}
\end{figure}

The results are shown in Figs.~\ref{fig:s5_design1_3d} and~\ref{fig:s5_design1_m1_3d}. 
For the proposed two-subarray design in Fig.~\ref{fig:s5_design1_3d}, the reconstructed target remains on the road, which indicates correct localization in this noiseless example. 
In contrast, when only the subarray generated by $\mathbf M_1$ is used, the reconstructed target in Fig.~\ref{fig:s5_design1_m1_3d} is shifted away from the road. The similar ambiguity also occurs when only the subarray generated by $\mathbf M_2$ is used. 
These results visually confirm that using the two subarrays jointly enlarges the unambiguous range and leads to correct target localization in this example.

\section{Conclusion}\label{s6}

In this paper, we studied moving target SAR imaging with planar antenna arrays under the MD-CRT framework. For a target moving over a three-dimensional terrain (or in air), we formulated the motion-induced cross-range shift and the target height in a unified two-dimensional form, and showed that, after 2D-DFT processing across a planar array, they admit a vector remainder representation. Starting from a general 2D-DFT matrix modulus formulation, we further showed that the associated 2D-DFT matrix modulus plays an essential role in the two-dimensional setting, since it affects the propagation of vector remainder errors to the desired parameter vector. Under a fixed array geometry and antenna number constraint, we derived an optimal construction of this matrix that minimizes the error amplification factor.

Based on this formulation, we developed a multi-subarray ambiguity resolution framework that enlarges the unambiguous range of a single planar array and enables the joint recovery of the motion-induced shift and the target height. To address practical estimation errors, we introduced an approximate 2D-DFT peak model for non-integer frequency vectors, and incorporated robust MD-CRT into the proposed framework. We then established sufficient conditions together with explicit reconstruction error bounds in both noiseless and noisy settings. Numerical results further illustrated that using two planar subarrays jointly yields a larger unambiguous range than using a single planar array alone. Detailed planar subarray designs and more detailed simulations on robustness are presented in Part II of this two-part paper \cite{part2}, where we further show that recovery performance depends not only on the number of antennas but also on the array geometry.

\end{document}